\documentclass[aip,jmp,amsmath,amssymb,reprint, onecolumn]{revtex4-1}

\usepackage{graphicx}
\usepackage{rotating}
\usepackage{multirow}
\usepackage{pgfplots}
\usepackage{adjustbox}
\usepackage{tikz,pgfplots}
\usepackage{dcolumn}
\usepackage{bm}
\usepackage{amsmath,amssymb,amsfonts}
\usepackage[utf8]{inputenc}
\usepackage[T1]{fontenc}
\usepackage{mathptmx}
\usepackage{amsthm}
\newtheorem*{theorem*}{Theorem}
\newtheorem{theorem}{Theorem}
\newtheorem{lemma}{Lemma}
\newtheorem{proposition}{Proposition}

\theoremstyle{definition}

\theoremstyle{remark}
\newtheorem{remarkx}{Remark}
\usepackage{etoolbox}
\usepackage{upgreek}
\usepackage{siunitx}
\usepackage{booktabs}
\usepackage{makecell}  
\usepackage{mathrsfs}
\usepackage[mathcal]{eucal}
\renewcommand{\epsilon}{\upepsilon}

\usepackage{array}
\usepackage{setspace}
\usepackage{color}
\definecolor{red_n}{rgb}{1.0, 0.0, 0.0}
\definecolor{blue_n}{rgb}{0.0, 0.0, 1.0}
\definecolor{orange_n}{RGB}{255.0, 127.0, 0.0}
\definecolor{purple_n}{RGB}{128.0, 0.0, 128.0}
\definecolor{dark_green}{rgb}{0.0, 0.5, 0.0}
\makeatletter
\def\@email#1#2{%
 \endgroup
 \patchcmd{\titleblock@produce}
  {\frontmatter@RRAPformat}
  {\frontmatter@RRAPformat{\produce@RRAP{*#1\href{mailto:#2}{#2}}}\frontmatter@RRAPformat}
  {}{}
}%
\makeatother
\begin{document}

\title{Transverse modulation in electrovac Brinkmann pp‑waves: Maxwell consistency and curvature universality}
\author{Galin S. Valchev}
 \altaffiliation[Also at ]{University Hospital Sofiamed, G. M. Dimitrov blvd., 1797 Sofia, Bulgaria}
 \email{gvalchev@imbm.bas.bg}
\affiliation{Institute of Mechanics, Bulgarian Academy of Sciences, Acad. G. Bonchev str., Sofia 1113, Bulgaria}%

\date{\today}

\begin{abstract}
Electrovac pp--waves provide an exact arena in which co--propagating gravitational and electromagnetic radiation share a common Brinkmann profile. While the general Einstein--Maxwell pp--wave family is well understood, it is tempting in applications (e.g.\ lensing or scattering) to ``modulate'' a plane electromagnetic wave by a weak transverse envelope $1+\gamma f(x,y)$. In this work we show that such a prescription is generically incompatible with the source--free Maxwell equations within the aligned null pp--wave ansatz. 

We give a minimal, polarization--agnostic gauge completion of a transversely modulated plane--wave potential and prove a cancellation theorem: for fields that remain null and aligned ($F_{xy}=0$, no $v$--dependence), the Maxwell equations force the physical transverse field $F_{ui}$ to be curl--free and divergence--free in the transverse plane, i.e.\ locally the gradient of a harmonic scalar. Under standard regularity/decay (or zero--mode) conditions that exclude additional harmonic transverse modes, all $\mathcal O(\gamma)$ dependence on an arbitrary smooth profile $f(x,y)$ drops out of $F_{ui}$ and hence of the electrovac source $T_{uu}$. Consequently, to first order in $\gamma$ the electromagnetic contribution to the Brinkmann profile is universal: it consists of the familiar cycle--averaged isotropic $r^2$ term together with a purely oscillatory isotropic correction at frequency $2\omega$ that is present only for non--circular polarization. 

We isolate the residual harmonic freedom corresponding to genuine Maxwell--admissible transverse structure and relate it to the standard holomorphic/harmonic data of electrovac pp--waves. By Kerr--Schild linearity we further superpose an arbitrary co--propagating vacuum gravitational pp--wave and provide a transparent mapping between TT--gauge strain and Brinkmann amplitudes. Our results clarify which transverse ``envelopes'' are physical in electrovac pp--waves and identify the minimal extensions (currents, non--null fields, or more general Kundt/gyraton geometries) required to model localized beams.
\end{abstract}

\maketitle
\section{\label{sec:intro}Introduction}
Plane–fronted waves with parallel rays (pp‑waves) constitute a distinguished class of exact solutions of general relativity that model freely propagating radiation. They are characterized by a covariantly constant null vector \(k^\mu\) (\(\nabla_\nu k^\mu=0\), \(k^\mu k_\mu=0\)), so that the wavefronts (surfaces of constant retarded time \(u\)) are flat and the null rays remain parallel. In Brinkmann coordinates all curvature is encoded in a single transverse profile \(h(u,x,y)\), which makes these spacetimes an analytically transparent arena for isolating genuine tidal effects, superposing aligned fields, and assessing how matter responds to propagating spacetime distortions.\cite{Brinkmann1925,Griffiths2009,Stephani2003}

In Einstein--Maxwell theory, aligned \emph{null} electromagnetic fields co--propagating with a pp--wave source the same Brinkmann profile through the single Ricci component \(R_{uu}\).\cite{Bonnor1969,Hassaine2007,Griffiths2009}
The resulting electrovac pp--waves are classical and well understood: Maxwell’s equations constrain the admissible transverse dependence of the electromagnetic field, and the Einstein equations reduce to a transverse Poisson problem for \(h\) with source proportional to the null energy flux \(T_{uu}\).\cite{Griffiths2009,Stephani2003}
Throughout, by \emph{electrovac} we mean the source-free Einstein--Maxwell system (no charges/currents $J^\mu=0$) with no matter other than the electromagnetic field.
A further simplification is the Kerr--Schild nature of Brinkmann pp--waves, which renders the field equations linear in \(h\) within this class and allows straightforward superposition of an arbitrary vacuum pp--wave sector.\cite{Ehlers1962,Griffiths2001,DebneyKerrSchild1969}

In many physical situations one would like to model radiation that is nearly planar but not perfectly homogeneous across the wavefront: finite aperture, weak focusing/defocusing, gravitational lensing magnification and shear, or scattering through a turbulent medium naturally suggest a slowly varying transverse modulation of the field amplitude. A common heuristic is to start from a plane wave and multiply the potential (or fields) by a weak envelope factor \(1+\gamma f(x,y)\), with \(0<|\gamma|\ll1\) and \(f\) smooth.
A basic and important question is then whether such a prescription can be implemented within the \emph{source--free} Einstein--Maxwell pp--wave family, and—if it can—what new curvature it produces.

The central result of this paper is that, within the aligned \emph{null} pp--wave ansatz, an arbitrary smooth transverse envelope is \emph{not} a freely specifiable electrovac datum. More precisely: for \(v\)-independent, aligned null fields with \(F_{xy}=0\), Maxwell’s equations enforce that the transverse field \(F_{ui}\) is both divergence--free and curl--free in the \((x,y)\) plane, hence locally the gradient of a harmonic scalar. As a consequence, generic localized or non-harmonic modulations \(f(x,y)\) cannot survive as physical \(\mathcal O(\gamma)\) corrections to the Maxwell tensor without introducing additional harmonic modes or leaving electrovac. In a systematic weak--modulation expansion we show explicitly that once Maxwell’s equations are restored by a minimal gauge completion, all \(\mathcal O(\gamma)\) dependence on \(f(x,y)\) cancels from \(F_{ui}\) (and thus from \(T_{uu}\)) under standard regularity/decay (or zero--mode) conditions. The electromagnetic contribution to the Brinkmann profile is therefore \emph{universal} at \(\mathcal O(\gamma)\): it contains the familiar isotropic cycle--averaged \(r^2\) term and, for non--circular polarization, an isotropic oscillatory correction at frequency \(2\omega\). Any genuinely transverse structure compatible with source--free Maxwell theory corresponds to residual harmonic data, in agreement with the standard characterization of electrovac pp--waves.\cite{Griffiths2009,Stephani2003}

Concretely, we:
(i) formulate and prove a Maxwell--consistency theorem for weak transverse ``envelope'' modulations of an aligned plane electromagnetic wave, providing an explicit minimal gauge completion valid for arbitrary smooth \(f(x,y)\) and identifying the residual harmonic freedom;
(ii) establish a cancellation/universality result for the electrovac source \(T_{uu}\) and hence for the electromagnetic contribution to the Brinkmann profile \(h\) at \(\mathcal O(\gamma)\) under natural boundary/regularity assumptions, and give the resulting polarization--resolved DC/AC decomposition;
(iii) compute the corresponding curvature in closed form and document the near--axis tidal content, emphasizing which pieces are universal and which require genuinely harmonic Maxwell data; and
(iv) superpose an arbitrary co--propagating vacuum gravitational pp--wave by Kerr--Schild linearity and provide a transparent map between the TT--gauge strain and Brinkmann amplitudes in appropriately rescaled coordinates.\cite{Ehlers1962,Griffiths2001,Griffiths2009,DebneyKerrSchild1969}

Our analysis is restricted to aligned \emph{null} Einstein--Maxwell pp--waves (electrovac) in Brinkmann form and to weak transverse modulations organized by a bookkeeping parameter \(\gamma\).
The universality result does \emph{not} imply that transverse structure is impossible in electrovac pp--waves; rather, it clarifies that the admissible transverse dependence is harmonic (equivalently, captured by the standard holomorphic/harmonic Maxwell data), and that generic localized envelopes must be supported by physics beyond electrovac pp--waves—e.g.\ external currents, non--null field components, or more general Kundt/gyraton geometries appropriate to finite beams.

Section~\ref{sec:BrinkmanIntro} reviews Brinkmann pp--wave kinematics and the Einstein--Maxwell reduction. Section~\ref{subsec:arbPolEMwaveweakf} formulates the weakly modulated potential ansatz and presents the Maxwell--consistency/gauge--completion theorem, including the harmonic obstruction and residual freedom. Sections~\ref{Tuusource}--\ref{TDC} derive the resulting (universal) electrovac source \(T_{uu}\), its cycle--averaged and oscillatory decomposition, and the corresponding Brinkmann profile to \(\mathcal O(\gamma)\). Section~\ref{Curvecont} analyses the curvature and near--axis tidal content. Section~\ref{GWsuperpos} adds an arbitrary vacuum gravitational pp--wave and connects Brinkmann amplitudes to TT--gauge strain. Appendices collect technical identities and gauge/boundary condition remarks.

\section{The Brinkmann Electrovacuum Metric: Properties and Physical Interpretation}\label{sec:BrinkmanIntro}
In a standard coordinate representation—known as the Brinkmann form—the metric reads:
\begin{equation}\label{Brinkmann_metric}
\mathrm{d}s^{2} = -2\,\mathrm{d}u\,\mathrm{d}v + \mathrm{d}x^{2} + \mathrm{d}y^{2} + h(u,x,y)\,\mathrm{d}u^{2},
\end{equation}
where \(u=(ct-z)/\sqrt{2}\) (retarded time), \(v=(ct+z)/\sqrt{2}\) (an affine parameter along the rays), and \(x,y\) span the transverse plane.\cite{Brinkmann1925,Ehlers1962,Kundt1961,Griffiths2009,Stephani2003} The function \(h(u,x,y)\) encodes the wave profile and polarization. A common vacuum choice is quadratic in \(x,y\), e.g.
\(h(u,x,y)=a(u)\,(x^{2}-y^{2})+2\,b(u)\,x y\),
representing the two independent polarizations.

All nontrivial curvature of Eq. \eqref{Brinkmann_metric} resides in the \(h\)-term. The only independent Riemann components are
\(
R_{uiuj} = -\tfrac{1}{2}\,\partial_i\partial_j h
\)
\((i,j\in\{x,y\})\), and the Ricci tensor reduces to:
\begin{equation}
R_{uu} = -\tfrac{1}{2}\,\Delta_{\perp} h, \qquad
\Delta_{\perp} \equiv \partial_x^2 + \partial_y^2,
\end{equation}
as is standard for pp--wave metrics in Brinkmann form.\cite{Griffiths2009,Stephani2003}

The Brinkmann metric can also be written in Kerr–Schild form:
\begin{equation}\label{KerrSchild}
g_{\mu\nu} = \eta_{\mu\nu} + h(u,x,y)\,k_{\mu}k_{\nu},
\end{equation}
where \(\eta_{\mu\nu}=\mathrm{diag}(-1,+1,+1,+1)\) is the Minkowski metric in Cartesian coordinates \(x_{\rm M}^{\mu}=(ct,x,y,z)\) and \(k_\mu\,\mathrm{d}x_{\rm M}^\mu=\mathrm{d}u\). Here \(k^\mu\) is null with respect to both \(g_{\mu\nu}\) and \(\eta_{\mu\nu}\) and is geodesic (\(k^\nu\nabla_\nu k^\mu=0\)). The Kerr–Schild structure renders the field equations linear in \(h\) for this class, greatly simplifying analyses and clarifying the superposition with other aligned fields.\cite{Ehlers1962,Griffiths2001,Griffiths2009,DebneyKerrSchild1969}
In particular, we will consider aligned electromagnetic fields with \(A_v=0\) and no \(v\)-dependence, so that the only potentially nonzero components are \(F_{ui}\) and \(F_{ij}\) on the transverse plane. In the null/aligned sector relevant to electrovac pp--waves we impose \(F_{xy}=0\), and Maxwell's equations then enforce that \(F_{ui}\) is simultaneously divergence--free and curl--free in \((x,y)\), hence locally the gradient of a harmonic scalar. This harmonic restriction is the key consistency input behind the universality/cancellation results established in Secs.~\ref{subsec:arbPolEMwaveweakf}--\ref{TDC}.

Overall, the Brinkmann framework affords a concise and versatile setting for studying gravitational radiation, its polarization content, and its interplay with matter and gauge fields in curved spacetime.\cite{Kundt1961,Ehlers1962,Griffiths2009,Stephani2003} A concise geometric classification (Kundt/pp--wave structure, Petrov type, VSI property) is summarized in Appendix~\ref{app:geom-class}.

\subsection{Coupling between a null electromagnetic wave and a pp-wave geometry}\label{subsec:EMppwacecoup}
For a suitably aligned class of fields, the Einstein–Maxwell system:
\begin{equation}\label{EinsteinMaxwell}
\begin{split}
&R_{\mu\nu}-\tfrac{1}{2}R\,g_{\mu\nu}=\frac{8\pi G}{c^{4}}\, T_{\mu\nu}^{(\mathrm{EM})},\qquad
\nabla_{\nu}F^{\mu\nu}=0,\qquad F_{[\mu\nu;\lambda]}=0,\\
&T_{\mu\nu}^{(\mathrm{EM})}\equiv \frac{1}{\mu_{0}}\!\left(F_{\mu\alpha}F_{\nu}{}^{\alpha}-\tfrac{1}{4}g_{\mu\nu}F_{\alpha\beta}F^{\alpha\beta}\right),
\end{split}
\end{equation}
admits closed-form solutions within the Brinkmann (pp-wave) ansatz Eq. \eqref{Brinkmann_metric}. Here \(G\) is Newton’s constant and \(\mu_{0}\) the vacuum permeability (SI).\cite{Wald1984,Misner1973} 

We begin with the homogeneous plane--wave specialization, which already reproduces the universal isotropic \(r^2\) Brinkmann profile. In later sections we allow weak transverse dependence and show that Maxwell consistency forces it to reduce to harmonic (or pure-gauge) data at \(\mathcal O(\gamma)\).
We take an electromagnetic potential aligned with the wave vector:
\begin{equation}\label{Acurved}
A_{\mu}\,\mathrm{d}x_{\rm Br}^{\mu}=A_{x}(u)\,\mathrm{d}x + A_{y}(u)\,\mathrm{d}y,
\end{equation}
i.e., choose the gauge \(A_{v}=A_{u}=0\) with \(A_{x},A_{y}\) depending only on \(u\), and $x_{\rm Br}^{\mu}=(u,v,x,y)$. The only nonvanishing components of the field tensor are then:
\begin{equation}\label{Fnonzero}
F_{ux}=\partial_{u}A_{x}(u),\qquad F_{uy}=\partial_{u}A_{y}(u),
\end{equation}
and the invariants \(F_{\mu\nu}F^{\mu\nu}\) and \(F_{\mu\nu}{^{\star}F}^{\mu\nu}\) vanish, so the field is null.\cite{Bonnor1969,Griffiths2009}

The homogeneous plane--wave gauge choice Eqs. \eqref{Acurved}--\eqref{Fnonzero} is a convenient special case.
More generally, one may allow:
\begin{equation}\label{eq:aligned-general-A}
A = A_u(u,x,y)\,\mathrm d u + A_i(u,x,y)\,\mathrm d x^i,
\qquad A_v=0,\qquad \partial_v A_\mu=0,
\end{equation}
so that
\(
F_{ui}=\partial_u A_i-\partial_i A_u
\)
and
\(
F_{ij}=\partial_i A_j-\partial_j A_i
\).
Since \(\sqrt{-g}=1\) for the Brinkmann metric, the source--free Maxwell equations reduce to
\(\partial_\nu F^{\mu\nu}=0\).
With the aligned ansatz Eq. \eqref{eq:aligned-general-A} this gives, for the relevant components:
\begin{equation}\label{eq:Maxwell-transverse}
\partial_i F_{ui}=0,
\qquad
\partial_j F_{ij}=0,
\qquad i,j\in\{x,y\}.
\end{equation}
Imposing the \emph{null/aligned} condition \(F_{xy}=0\) makes the second equation in Eq. \eqref{eq:Maxwell-transverse} identically satisfied, while the Bianchi identity \(F_{[\mu\nu;\lambda]}=0\) (equivalently \(dF=0\)) yields the transverse curl constraint:
\begin{equation}\label{eq:curlfree}
\partial_i F_{uj}-\partial_j F_{ui}=0.
\end{equation}
On a simply connected transverse patch, Eq. \eqref{eq:curlfree} implies the existence of a scalar potential \(\Phi(u,x,y)\) such that:
\begin{equation}\label{eq:Fui-harmonic}
F_{ui}=\partial_i \Phi(u,x,y),
\qquad
\Delta_\perp \Phi(u,x,y)=0,
\end{equation}
where the harmonic condition follows from \(\partial_i F_{ui}=0\) in Eq. \eqref{eq:Maxwell-transverse}.
Thus, in electrovac Brinkmann pp--waves with aligned null Maxwell fields, the admissible transverse structure of the electromagnetic wave is not an arbitrary envelope but harmonic data (modulo gauge).
The homogeneous plane wave Eq. \eqref{Fnonzero} corresponds to \(\Phi(u,x,y)=a'_x(u)\,x+a'_y(u)\,y\), whose gradient is constant in \((x,y)\).
This harmonic restriction is the starting point for our Maxwell--consistency analysis of weak transverse ``envelopes'' in Sec.~\ref{subsec:arbPolEMwaveweakf}.

With this ansatz, the (electromagnetic) stress–energy has a single nonzero component:
\begin{equation}\label{Tuu}
T_{uu}^{(\mathrm{EM})}=\frac{1}{\mu_{0}}\!\left(F_{ux}F_{u}{}^{x}+F_{uy}F_{u}{}^{y}\right)
=\frac{1}{\mu_{0}}\!\Big[(\partial_{u}A_{x})^{2}+(\partial_{u}A_{y})^{2}\Big],
\end{equation}
using \(g^{xx}=g^{yy}=1\) in the Brinkmann metric.\cite{Bonnor1969} Since the only nontrivial curvature is \(R_{uu}\), the Einstein equations reduce to:
\begin{equation}\label{Huxyeq}
-\tfrac{1}{2}\,\Delta_{\perp} h(u,x,y)=\frac{8\pi G}{c^{4}}\,T_{uu}^{(\mathrm{EM})},
\qquad\Longrightarrow\qquad
\Delta_{\perp} h(u,x,y)=-\frac{16\pi G}{c^{4}\mu_{0}}\Big[(\partial_{u}A_{x})^{2}+(\partial_{u}A_{y})^{2}\Big],
\end{equation}
in agreement with the general pp–wave analysis.\cite{Griffiths2009,AyonBeato2007b} Thus, in an electrovac pp‑wave the \emph{transverse Laplacian} of the Brinkmann profile is sourced directly by the instantaneous energy flux of the aligned electromagnetic wave. The bracketed term is nonnegative in SI units, so that $\Delta_{\perp} h\le 0$ for positive energy density. A genuine gravitational contribution still requires transverse quadratic dependence of \(h\) on \((x,y)\) -- any term depending only on \(u\) is pure gauge and carries no curvature.\cite{Rosen1937,Griffiths2009,Stephani2003}

\section{Arbitrarily polarized electromagnetic wave with weak transverse inhomogeneity}%
\label{subsec:arbPolEMwaveweakf}
We consider a monochromatic plane electromagnetic wave propagating along $+z$ (the Brinkmann null direction $k=\partial_v$), with a weak, slowly varying transverse inhomogeneity. Let:
\begin{equation}
\theta(u)=\frac{\omega}{c}\,u+\phi_0,
\label{eq:theta_def}
\end{equation}
and parametrize the polarization by the ellipse orientation $\psi$, the ellipticity angle $\chi\in[-\tfrac{\pi}{4},\tfrac{\pi}{4}]$ (linear -- $\chi=0$, circular -- $|\chi|=\tfrac{\pi}{4}$), and the handedness $\sigma=\pm1$ (right/left). Here \(\omega\) is the angular frequency conjugate to the Brinkmann retarded time \(u\).
Since \(u=(ct-z)/\sqrt2\), a plane wave with inertial (laboratory) angular frequency \(\omega_{\rm phys}\)
and phase \(\omega_{\rm phys}(t-z/c)+\phi_0\) corresponds to
$\theta(u)=(\sqrt2\,\omega_{\rm phys}/c)\,u+\phi_0,$
so that in Eq. \eqref{eq:theta_def} one has \(\omega=\sqrt2\,\omega_{\rm phys}\).
All ``frequencies'' below are understood in this \(u\)-conjugate sense unless stated otherwise. 

A convenient \emph{real} representation of the transverse components of the vector potential in Brinkmann coordinates is:
\begin{equation}\label{Aarbreal}
\begin{split}
&A_x(u,x,y)=\frac{\sqrt{2}\,E_0}{\omega}\,\big[1+\gamma\,f(x,y)\big]\;\Big(\cos\chi\,\cos\theta\,\cos\psi-\sigma\,\sin\chi\,\sin\theta\,\sin\psi\Big),\\
&A_y(u,x,y)=\frac{\sqrt{2}\,E_0}{\omega}\,\big[1+\gamma\,f(x,y)\big]\;\Big(\cos\chi\,\cos\theta\,\sin\psi+\sigma\,\sin\chi\,\sin\theta\,\cos\psi\Big),
\end{split}
\end{equation}
with $0<|\gamma|\ll1$ and $E_0$ setting the overall scale. The factor $\sqrt{2}$ is chosen so that $\langle(\partial_u A_x)^2+(\partial_u A_y)^2\rangle_u=E_0^2/c^2$ for \emph{any} polarization; for circular polarization the sum is also independent of $u$.\cite{Born1999,Jackson1999} For brevity we write $r^2\equiv x^2+y^2$.

The multiplicative ``envelope'' ansatz Eqs. \eqref{Aarbreal} with \(A_u=0\) is a \emph{naive} way to encode a weak
transverse modulation. Indeed, with \(A_u=0\) one has
\(F_{ui}=\partial_u A_i\propto a_i'(u)\,[1+\gamma f(x,y)]\), so the \(v\)-component Maxwell equation
\(\partial_i F_{ui}=0\) fails generically because \(\partial_i F_{ui}\propto \mathbf a'(u)\!\cdot\!\nabla f\neq0\)
for an arbitrary \(f\).

More importantly for the electrovac pp--wave sector of interest here, we restrict to \emph{aligned null} Maxwell fields,
i.e. \(A_v=0\), \(\partial_v A_\mu=0\), and \(F_{xy}=0\). In that sector Maxwell plus Bianchi enforce that
\(F_{ui}\) must be both divergence--free and curl--free on the \((x,y)\) plane, hence locally:
\begin{equation}\label{eq:harmonic-constraint-reminder}
F_{ui}=\partial_i\Phi(u,x,y),\qquad \Delta_\perp\Phi(u,x,y)=0,
\end{equation}
so the admissible transverse structure is \emph{harmonic} (modulo gauge/boundary conditions).
The theorem below gives an explicit Maxwell--consistent completion of the naive modulated potential and makes
the surviving harmonic freedom manifest.
 
\begin{theorem}\label{thm:gauge-completion}
Let the zeroth--order (homogeneous) plane--wave potential be:
\begin{equation}\label{zeroplanewave}
A^{(0)} = a_x(u)\,\mathrm d x+a_y(u)\,\mathrm d y,
\qquad a_i\in C^{1}(u),\qquad A^{(0)}_{u}=A^{(0)}_{v}=0,
\end{equation}
and fix any smooth transverse profile \(f\in C^\infty(\mathbb R^2)\) and a small parameter \(0<|\gamma|\ll1\).
Choose a particular solution \(w(x,y)\) of the Poisson equation:
\begin{equation}\label{eq:wPoisson}
\Delta_\perp w = f,
\end{equation}
and let \(\psi_{\rm H}(u,x,y)\) be any (possibly \(u\)-dependent) \emph{harmonic} scalar on the transverse plane:
\begin{equation}\label{eq:psiH}
\Delta_\perp \psi_{\rm H}(u,x,y)=0\qquad\text{for each fixed }u.
\end{equation}
Define the \(\mathcal O(\gamma)\) corrected potential by:
\begin{equation}\label{eq:AnsatzA-option1}
\begin{aligned}
A_u(u,x,y) &= \gamma\,\mathbf a'(u)\!\cdot\!\nabla w(x,y),\\
A_i(u,x,y) &= a_i(u)\;+\;\gamma\,\partial_i\!\Big[\mathbf a(u)\!\cdot\!\nabla w(x,y)\;+\;\psi_{\rm H}(u,x,y)\Big],
\qquad i=x,y,\\
A_v(u,x,y)&=0,
\end{aligned}
\end{equation}
where \(\mathbf a=(a_x,a_y)^{\top}\), \(\nabla=(\partial_x,\partial_y)^{\top}\).

Then:
\begin{enumerate}
\item The field is aligned and null:
\begin{equation}\label{eq:null-aligned}
F_{xy}=0,\qquad F_{\mu\nu}F^{\mu\nu}=F_{\mu\nu}{}^{\star}\!F^{\mu\nu}=0.
\end{equation}
\item Maxwell's equations \(\nabla_\nu F^{\mu\nu}=0\) and the Bianchi identity hold through \(\mathcal O(\gamma)\).
\item The physical radiative components can be written in the harmonic form:
\begin{equation}\label{eq:Fui-harmonic-form}
F_{ui}(u,x,y)=\partial_i\Phi(u,x,y),
\qquad
\Phi(u,x,y)=\mathbf a'(u)\!\cdot\!\mathbf x\;+\;\gamma\,\partial_u\psi_{\rm H}(u,x,y),
\qquad
\Delta_\perp\Phi=0,
\end{equation}
where \(\mathbf x=(x,y)^{\top}\).
In particular, \emph{all dependence on the arbitrary profile \(f(x,y)\) cancels out of \(F_{ui}\) at order \(\gamma\)}.
Nontrivial transverse structure survives only through the freely prescribable harmonic datum \(\psi_{\rm H}\).
\end{enumerate}

If, on \(\mathbb R^2\), one imposes that \(\psi_{\rm H}\) is bounded (or decays) as \(r\to\infty\), then \(\psi_{\rm H}\) is constant
and the \(\mathcal O(\gamma)\) correction vanishes: \(F_{ui}=a'_i(u)\).
\end{theorem}

\begin{proof}
For the Brinkmann metric \(\sqrt{-g}=1\), and with \(A_v=0\) and no \(v\)-dependence the Maxwell equations reduce to
\(\partial_\nu F^{\mu\nu}=0\) (the connection term vanishes because all nonzero Christoffels carry a lower \(u\) index
while \(F^{u\nu}=0\) for the aligned ansatz).

From Eq. \eqref{eq:AnsatzA-option1} we have \(A_i=a_i+\gamma\,\partial_i\Psi\) with
\(\Psi:=\mathbf a\!\cdot\!\nabla w+\psi_{\rm H}\). Hence:
\begin{equation}
F_{xy}=\partial_x A_y-\partial_y A_x
=\gamma\big(\partial_x\partial_y\Psi-\partial_y\partial_x\Psi\big)=0,
\end{equation}
which implies the null invariants in Eq. \eqref{eq:null-aligned} for this aligned class.

Next, we write:
\begin{equation}
F_{ui}=\partial_u A_i-\partial_i A_u
=a'_i(u)+\gamma\,\partial_i\!\Big[\mathbf a'(u)\!\cdot\!\nabla w\;+\;\partial_u\psi_{\rm H}\Big]
-\gamma\,\partial_i\!\big[\mathbf a'(u)\!\cdot\!\nabla w\big]
=a'_i(u)+\gamma\,\partial_i\partial_u\psi_{\rm H}.
\end{equation}
This is precisely Eq. \eqref{eq:Fui-harmonic-form} with
\(\Phi=\mathbf a'(u)\!\cdot\!\mathbf x+\gamma\,\partial_u\psi_{\rm H}\).
Since \(\Delta_\perp(\mathbf a'\!\cdot\!\mathbf x)=0\) and \(\Delta_\perp\psi_{\rm H}=0\),
it follows that \(\Delta_\perp\Phi=0\).

Finally, the \(v\)-component Maxwell equation is \(\partial_i F_{ui}=0\), which holds because
\(\partial_i F_{ui}=\Delta_\perp\Phi=0\).
The \(x,y\)-components reduce to \(\partial_j F_{ij}=0\), which are identically satisfied because \(F_{xy}=0\).
The Bianchi identity holds identically since \(F=\mathrm dA\).
\end{proof}

\begin{remarkx}\label{rem1}
A particularly convenient choice is to take:
\begin{equation}\label{ggradw}
g=\nabla w,\qquad \Delta_\perp w=f,
\end{equation}
so that the symmetric gradient
\(
S:=\tfrac12\!\left(\nabla g+\nabla g^{\!\top}\right)=\nabla\nabla w
\)
is the \emph{Hessian} of $w$. This has two immediate consequences used later:
\begin{equation}\label{trS}
\mathrm{tr}\,S=\Delta_\perp w=f,\qquad
S_{ij}=\partial_i\partial_j w=\partial_j\partial_i w .
\end{equation}
The completion for $A_u$ becomes:
\begin{equation}\label{Auw}
A_u(u,x,y)=\gamma\Big[a'_x(u)\,\partial_x w(x,y)+a'_y(u)\,\partial_y w(x,y)\Big],
\end{equation}
and its transverse derivatives are
\(
\partial_i A_u=\gamma\,[a'_x\,\partial_i\partial_x w+a'_y\,\partial_i\partial_y w]
=\gamma\,S_{ij}\,a'_j
\)
(with the obvious index convention). In particular:
\begin{equation}\label{Auf}
\Delta_\perp A_u
=\gamma\,[a'_x\,\Delta_\perp (\partial_x w)+a'_y\,\Delta_\perp (\partial_y w)]
=\gamma\,[a'_x\,\partial_x f+a'_y\,\partial_y f].
\end{equation}

In the null/aligned completion of Theorem~\ref{thm:gauge-completion}, the \(w\)-dependent pieces cancel in
\(F_{ui}=\partial_u A_i-\partial_i A_u\), so \(w\) (hence \(f\)) affects only the gauge potential, while the physical
field strength depends only on the harmonic datum \(\psi_{\rm H}\).
\end{remarkx}

\begin{remarkx}
Theorem~\ref{thm:gauge-completion} shows that in the aligned null sector the physical degrees of freedom reduce to a harmonic scalar.
One may alternatively obtain Maxwell-consistent potentials by solving transverse Poisson problems in a chosen gauge (e.g. Coulomb),
but the resulting field strength is the same and depends only on the harmonic datum.
\end{remarkx}

In the remainder we adopt the aligned null Maxwell completion of Theorem~\ref{thm:gauge-completion}.
Given a prescribed transverse profile \(f(x,y)\), choose any particular solution \(w\) of
\(\Delta_\perp w=f\) [Eq.~\eqref{eq:wPoisson}]. The remaining transverse freedom is encoded in a
harmonic scalar \(\psi_{\rm H}(u,x,y)\) satisfying \(\Delta_\perp\psi_{\rm H}=0\) for each fixed \(u\)
[Eq.~\eqref{eq:psiH}]. The physical radiative components are then
\(F_{ui}=\partial_i\Phi\) with \(\Delta_\perp\Phi=0\), i.e. the transverse structure is harmonic
[Eq.~\eqref{eq:Fui-harmonic-form}]. Under standard boundedness/decay conditions on \(\mathbb R^2\),
\(\psi_{\rm H}\) is constant and the \(\mathcal O(\gamma)\) correction to \(F_{ui}\) vanishes.

\subsection{Astrophysical motivation and the choice of \(f(x,y)\)}\label{Astromotiv}
Elliptically polarized radiation is ubiquitous in high-energy and radio
astrophysics. Intrinsic emission processes (e.g.\ synchrotron in ordered
magnetic fields, Zeeman-split maser lines) as well as propagation effects
through magnetized plasmas (Faraday rotation and conversion) naturally yield
general polarization states \((\psi,\chi,\sigma)\).\citep{Rybicki1979,Jackson1999,Melrose1980a,Melrose1980,Westfold1959,Jones1977,Sazonov1969,Wardle2003}
Moreover, realistic wavefronts are never perfectly planar across astrophysical
distances: finite source size, differential gravitational lensing (convergence
and shear), and interstellar scattering induce small transverse inhomogeneities
in phase and amplitude.\citep{Schneider1992,Petters2001,Blandford1986,Rickett1990,Narayan1989,Armstrong1995}

A common phenomenological way to encode such effects is to introduce a weak transverse
``envelope'' factor \(1+\gamma f(x,y)\), with \(0<|\gamma|\ll1\), multiplying the
homogeneous plane-wave potential [cf.\ Eq.~\eqref{Aarbreal}]. In the present work
this ansatz is used as a \emph{starting point} to test Maxwell consistency in the
strict electrovac Brinkmann pp--wave setting. The central message of the present Maxwell-consistent analysis
(Sec.~\ref{subsec:arbPolEMwaveweakf} and Sec.~\ref{Tuusource}) is that, within the aligned null sector
(\(A_v=0\), \(\partial_v A_\mu=0\), \(F_{xy}=0\)), the source--free Maxwell equations
force the physical radiative field \(F_{ui}\) to be both divergence--free and curl--free
on the transverse plane, hence locally the gradient of a harmonic scalar,
\(\Delta_\perp\Phi=0\) [cf.\ Eq.~\eqref{eq:harmonic-constraint-reminder} and
Theorem~\ref{thm:gauge-completion}]. Consequently, a generic smooth envelope \(f(x,y)\)
does \emph{not} survive as a physical \(\mathcal O(\gamma)\) transverse modulation of the
electrovac source \(T_{uu}\) (and thus of the Brinkmann profile \(h\)) once Maxwell’s
equations are enforced: any non-harmonic part is removed by the required completion,
and only genuinely harmonic transverse data can remain (subject to boundary/regularity
conditions).

Mathematically, we assume
\begin{equation}\label{assumptions}
f\in C^{\infty}(\mathbb R^2),\qquad
\|f\|_{L^\infty}\lesssim 1,\qquad
\varepsilon:=\frac{\lambda}{2\pi L}\ll1,
\end{equation}
where \(L\) is the characteristic transverse variation scale of \(f\) and
\(\lambda\) the carrier wavelength. The smallness parameter \(\varepsilon\) encodes the
separation of scales (slowly varying envelope) typically invoked in beam/envelope
modelling, while \(|\gamma|\ll1\) ensures weak modulation. In the strict electrovac
pp--wave analysis, however, the \emph{physical} transverse degrees of freedom are
ultimately restricted to harmonic data by Maxwell consistency; the role of \(f\) is
therefore to motivate and organize the completion and to isolate which parts of a
phenomenological modulation can be represented within the aligned null electrovac class.

From the electrovac viewpoint, the natural desiderata for transverse structure become:
(i) \emph{harmonic admissibility}: only transverse dependence compatible with
\(\Delta_\perp\Phi=0\) (and the chosen boundary conditions) can represent physical aligned null
Maxwell data at \(\mathcal O(\gamma)\);
(ii) \emph{centering} at the reference geodesic (no constant/dipole pieces in a local Taylor
expansion), so that the axis \(x=y=0\) corresponds to a natural reference worldline for the
near-axis curvature analysis in Sec.~\ref{Curvecont};
(iii) \emph{analytic tractability}: harmonic quadrupoles and higher harmonic multipoles provide
canonical local models of admissible transverse structure, and they also coincide with the standard
quadratic basis used to describe weak lensing shear and vacuum pp--wave polarizations.

\begin{remarkx}\label{rem:lens-to-f}
In geometric (eikonal) optics for a thin gravitational lens, the scaled lensing
potential $\psi_{\rm L}(\boldsymbol\theta)$ satisfies
$\boldsymbol\beta=\boldsymbol\theta-\nabla_\theta\psi_{\rm L}$, and the
\emph{Jacobi (amplification) matrix} is the $2\times2$ Jacobian:
\begin{equation}\label{Jacobimatrix}
\mathsf A(\boldsymbol\theta):=\frac{\partial\boldsymbol\beta}{\partial\boldsymbol\theta}
=\mathbb I-\nabla\nabla\psi_{\rm L}
= (1-\kappa)\,\mathbb I - \Gamma,
\qquad
\Gamma:=\begin{pmatrix}\gamma_1 & \gamma_2\\ \gamma_2 & -\gamma_1\end{pmatrix},
\end{equation}
with convergence $\kappa:=\tfrac12\,\Delta_\perp\psi_{\rm L}$ and shear components
$(\gamma_1,\gamma_2)$ defined by the traceless Hessian of $\psi_{\rm L}$.
\citep{Schneider1992,Blandford1986,Petters2001}
The scalar \emph{magnification} is:
\begin{equation}\label{scalarmag}
\mu(\boldsymbol\theta) \;=\; \frac{1}{\det\mathsf A}
\;=\; \frac{1}{(1-\kappa)^2-|\gamma|^2},
\qquad |\gamma|^2:=\gamma_1^2+\gamma_2^2.
\end{equation}
Here $\gamma_1$ and $\gamma_2$ denote the standard weak-lensing shear components (not the modulation parameter $\gamma$ used elsewhere).
For weak lensing ($|\kappa|,|\gamma|\ll1$),
\begin{equation}\label{weaklensing}
\mu^{1/2} \;=\; 1+\kappa + \mathcal O\!\left(\kappa^2,\,\kappa|\gamma|,\,|\gamma|^2\right),
\end{equation}
so to first order the \emph{isotropic} amplitude change is set by $\kappa$, while the
\emph{anisotropic} distortion is encoded in the traceless shear matrix $\Gamma$.

Across a small image-plane patch one may Taylor-expand any slowly varying multiplicative modulation
to quadratic order in \(\mathbf x=(x,y)\) about the center. The anisotropic (shear) part must be
built from a traceless quadratic form. The unique quadratic scalars with vanishing Laplacian are
the two \emph{harmonic quadrupoles}:
\(
(x^2-y^2)/2\) and \(x\,y,
\)
which rotate into each other under in-plane rotations. Combining them gives the canonical
shear-aligned quadratic:
\begin{equation}\label{eq:f-shear-canonical}
f_{\rm shear}(x,y)
\;=\; \frac{1}{2}\Big[\gamma_1\,(x^2-y^2) + 2\,\gamma_2\,x y\Big]
\;=\; \frac{1}{2}\,\mathbf x^{\!\top}\Gamma\,\mathbf x.
\end{equation}

\begin{lemma}\label{lem1}
Let $\Gamma$ be any real $2\times2$ traceless symmetric matrix.
Then $f(\mathbf x)=\tfrac12\,\mathbf x^{\!\top}\Gamma\,\mathbf x$ satisfies
$\Delta_\perp f=\operatorname{tr}\Gamma=0$ and
$\nabla\nabla f=\Gamma$.
\end{lemma}
\begin{proof} Direct differentiation gives
$\partial_i\partial_j f = \Gamma_{ij}$ and
$\Delta_\perp f = \operatorname{tr}\Gamma = 0$.
\end{proof}

Equation~\eqref{eq:f-shear-canonical} highlights why harmonic quadrupoles are natural test profiles:
in a locally weakly lensed wavefront, the leading anisotropic modulation is harmonic and aligned
with the shear principal axes. This dovetails with the electrovac pp--wave constraint
\(\Delta_\perp\Phi=0\): the shear quadrupole belongs to the lowest nontrivial harmonic sector that can
represent aligned null Maxwell transverse data (through the harmonic datum \(\psi_{\rm H}\) in
Theorem~\ref{thm:gauge-completion}) or, independently, the vacuum gravitational pp--wave sector which is
also governed by harmonic transverse profiles.

By contrast, the isotropic convergence contribution \(\propto r^2\) is \emph{not} harmonic and therefore
cannot be realized as a source--free aligned null Maxwell modulation within the strict Brinkmann electrovac
pp--wave class. In the Maxwell-consistent electrovac framework developed here this non-harmonic envelope content is precisely what is removed by
Maxwell consistency: it does not contribute to the electrovac source \(T_{uu}\) at \(\mathcal O(\gamma)\)
once the completion is imposed.
\end{remarkx}

In summary, the envelope factor \(1+\gamma f(x,y)\) is a flexible astrophysical \emph{motivation}, but the strict
electrovac Brinkmann pp--wave analysis shows that its generic (non-harmonic) transverse content does not backreact
on the metric at \(\mathcal O(\gamma)\). The physically admissible transverse structure in aligned null electrovac
pp--waves is harmonic (subject to boundary conditions), with harmonic quadrupoles providing the canonical local
example. Under the standard convention that excludes additional harmonic transverse modes on \(\mathbb R^2\),
the on-axis and near-axis curvature is universal and determined solely by the homogeneous electromagnetic energy
flux and the polarization-controlled oscillatory factor (Sec.~\ref{Curvecont}); modelling genuinely localized
beam envelopes or stochastic screens requires physics beyond the strict electrovac pp--wave setting (e.g.\ external
currents, non-null field components, or more general Kundt/gyraton geometries).

\subsection{Electromagnetic source \(T_{uu}\) to \(\mathcal O(\gamma)\)}\label{Tuusource}
For an aligned Maxwell field in the Brinkmann background (with \(A_v=0\) and no
\(v\)-dependence), the only stress--energy component that sources the profile
\(h(u,x,y)\) is:\cite{Bonnor1969,AyonBeato2007b,Griffiths2009}
\begin{align}
T^{(\mathrm{EM})}_{uu}(u,x,y)
&=\frac{1}{\mu_{0}}\Big(F_{ux}F_{u}{}^{x}+F_{uy}F_{u}{}^{y}\Big)
=\frac{1}{\mu_0}\sum_{i=x,y}\big(\partial_u A_i - \partial_i A_u\big)^2\nonumber\\
&=\frac{1}{\mu_0}\Big[(\partial_{u}A_{x})^{2}+(\partial_{u}A_{y})^{2}\Big]
-\frac{2}{\mu_0}\sum_{i=x,y}(\partial_{u}A_{i})(\partial_{i}A_{u})
+\mathcal O(\gamma^{2}),
\label{eq:Tuu_def}
\end{align}
where we used \(A_u=\mathcal O(\gamma)\) so that \((\partial_i A_u)^2=\mathcal O(\gamma^2)\)
can be neglected at the present order.

It is convenient to collect the homogeneous derivatives into the column vector:
\begin{equation}
\mathbf a'(u)\equiv\big[a'_x(u),\,a'_y(u)\big]^{\!\top},\qquad
\widehat{\mathbf a}'(u)\equiv\frac{c}{E_0}\,\mathbf a'(u),
\label{homderiv}
\end{equation}
so that \(\widehat{\mathbf a}'\) is dimensionless and its squared norm is proportional
to the instantaneous EM energy flux. We define the \emph{instantaneous polarization
dyadic}:
\begin{equation}
\label{eq:D_decomp}
\mathbf D(u)\;\equiv\;\widehat{\mathbf a}'(u)\,\widehat{\mathbf a}'(u)^{\!\top}
\;=\;\mathbf P(\psi,\chi)\;+\;\cos(2\theta)\,\mathbf Q_{c}(\psi,\chi)\;+\;\sin(2\theta)\,\mathbf Q_{s}(\psi,\chi,\sigma),
\end{equation}
so that \(\langle \widehat{\mathbf a}'{}^{\!\top}\widehat{\mathbf a}'\rangle_u=1\).
Writing \(\mathbf e_1=(\cos\psi,\sin\psi)^{\!\top}\) and \(\mathbf e_2=(-\sin\psi,\cos\psi)^{\!\top}\),
one finds the closed forms:
\begingroup
  \refstepcounter{equation}
  \label{eqQcQsdef}
  \xdef\eqQcQsdef{\theequation}
  \addtocounter{equation}{-1}
\endgroup
\begin{align}
\mathbf Q_{c}(\psi,\chi)&= \sin^{2}\!\chi\;\mathbf e_2\mathbf e_2^{\!\top}
-\cos^{2}\!\chi\;\mathbf e_1\mathbf e_1^{\!\top},
&
\mathrm{tr}\,\mathbf Q_c&=-\cos 2\chi,\tag{\eqQcQsdef.1}\label{eqQcQsdef-1}\\[2pt]
\mathbf Q_{s}(\psi,\chi,\sigma)&=-\frac{\sigma}{2}\,\sin 2\chi\;
\big(\mathbf e_1\mathbf e_2^{\!\top}+\mathbf e_2\mathbf e_1^{\!\top}\big),
&
\mathrm{tr}\,\mathbf Q_s&=0,\tag{\eqQcQsdef.2}\label{eqQcQsdef-2}
\end{align}
and, with the normalization of the \emph{homogeneous} plane wave (\(\gamma=0\)) in
Eqs.~\eqref{Aarbreal}, one has
\(
\mathrm{tr}\,\mathbf D(u)=\widehat{\mathbf a}'{}^{\!\top}\widehat{\mathbf a}'
=1-\cos 2\chi\,\cos(2\theta).
\)

In the aligned null sector (\(F_{xy}=0\)), Theorem~\ref{thm:gauge-completion} shows that
a naive transverse envelope can be completed so that the resulting physical field
strength takes the harmonic form:
\begin{equation}\label{eq:Fui-option1}
F_{ui}(u,x,y)=a'_i(u)\;+\;\gamma\,\partial_i\partial_u\psi_{\rm H}(u,x,y),
\qquad \Delta_\perp\psi_{\rm H}(u,\cdot)=0\ \ \text{for each fixed }u,
\end{equation}
and, crucially, \emph{all dependence on the arbitrary profile \(f(x,y)\) cancels from \(F_{ui}\) at
\(\mathcal O(\gamma)\)}. Expanding the quadratic contraction yields:
\begin{equation}
F_{ux}^2+F_{uy}^2
=\sum_{i=x,y}\Big[a'_i+\gamma\,\partial_i\partial_u\psi_{\rm H}\Big]^2
=a'_i a'_i+2\gamma\,a'_i\,\partial_i\partial_u\psi_{\rm H}
+\mathcal O(\gamma^2).
\label{eq:Fusq-option1}
\end{equation}
To express the correction compactly, define the dimensionless harmonic-completion vector and scalar:
\begingroup
  \refstepcounter{equation}
  \label{eqXidef}
  \xdef\eqXidef{\theequation}
  \addtocounter{equation}{-1}
\endgroup
\begin{align}
&\mathbf J_{\rm H}(u,x,y):=\frac{c}{E_0}\,\nabla\big[\partial_u\psi_{\rm H}(u,x,y)\big],\tag{\eqXidef.1}\label{eqXidef-1}\\[2pt]
&\Xi_{\rm H}(u,x,y):=\widehat{\mathbf a}'(u)^{\!\top}\mathbf J_{\rm H}(u,x,y)
=\frac{c^{2}}{E_0^{2}}\,a'_i(u)\,\partial_i\partial_u\psi_{\rm H}(u,x,y).\tag{\eqXidef.2}\label{eqXidef-2}
\end{align}
Using \(a'_i a'_i=\tfrac{E_0^2}{c^2}\,\mathrm{tr}\,\mathbf D(u)\), Eq.~\eqref{eq:Fusq-option1} becomes:
\begin{equation}
F_{ux}^2+F_{uy}^2
=\frac{E_0^2}{c^2}\Big\{\mathrm{tr}\,\mathbf D(u)+2\gamma\,\Xi_{\rm H}(u,x,y)\Big\}
+\mathcal O(\gamma^2),
\end{equation}
and therefore the electromagnetic source reads:
\begin{equation}
\label{eq:Tuu_instant_general-deriv}
T^{(\mathrm{EM})}_{uu}(u,x,y)
=\frac{E_0^2}{c^{2}\mu_0}\Big\{\mathrm{tr}\,\mathbf D(u)+2\gamma\,\Xi_{\rm H}(u,x,y)\Big\}
+\mathcal O(\gamma^{2}).
\end{equation}

\begin{remarkx}\label{rem:Xi-coulomb-cancel}
The only possible \(\mathcal O(\gamma)\) transverse dependence of \(T_{uu}\) in the aligned null electrovac
pp--wave sector is the harmonic-completion contribution \(\Xi_{\rm H}\), i.e. genuinely Maxwell-admissible
harmonic data. In many physical settings one excludes nontrivial harmonic complements:
for example, on \(\mathbb R^2\) any bounded (or decaying) harmonic \(\psi_{\rm H}\) is constant, so
\(\nabla(\partial_u\psi_{\rm H})\equiv0\) and hence \(\Xi_{\rm H}\equiv0\).
Similarly, on a periodic domain \(\mathbb T^2\) the only harmonic functions are constants (after fixing the
zero mode), so \(\Xi_{\rm H}\equiv0\).
Under such standard conditions the stress--energy is spatially homogeneous to \(\mathcal O(\gamma)\):
\begin{equation}\label{TuutrD}
T^{(\mathrm{EM})}_{uu}(u,x,y)
=\frac{E_0^2}{c^{2}\mu_0}\,\mathrm{tr}\,\mathbf D(u)+\mathcal O(\gamma^2),
\end{equation}
which is the source universality underlying the curvature universality results of Sec.~\ref{Curvecont}. A convenient particular solution of the Einstein equation
$\Delta_\perp h_{\rm EM}=-(16\pi G/c^4)\,T_{uu}^{(\mathrm{EM})}$ is then:
\begin{equation}\label{hEMsolution}
h_{\rm EM}(u,x,y)
=-\mathcal C\,\mathrm{tr}\,\mathbf D(u)\;r^2
+\mathcal O(\gamma^2),\ \text{with}\ \mathcal{C}=\frac{4\pi G E_0^{2}}{c^{6}\mu_{0}},
\end{equation}
as any additional solution of $\Delta_\perp h_{\rm vac}=0$ would represent an \emph{independent}
vacuum pp--wave superposed on the electrovac solution and is set to zero by convention.
\end{remarkx}

\subsection{Cycle-averaged profile}\label{Cycleavrgprof}
For many purposes it is natural to consider the cycle average of all
quantities in \(u\) (or, equivalently, in the phase \(\theta\)). For a
monochromatic wave this is the same as the time average over many optical
cycles, and we denote it by \(\langle\cdot\rangle_u\). At the level
of the polarization dyadic \(\mathbf D(u)=\widehat{\mathbf a}'(u)\widehat{\mathbf a}'(u)^{\!\top}\)
this yields the \emph{polarization projector}:
\begin{equation}
\mathsf P(\psi,\chi)\;\equiv\;\big\langle\widehat{\mathbf a}'\,\widehat{\mathbf a}'{}^{\!\top}\big\rangle_u
=\frac12\Big[\mathbb{I}_{2}+\cos 2\chi\; \mathsf M(\psi)\Big],\qquad
\mathsf M(\psi)=\begin{pmatrix}\cos 2\psi&\sin 2\psi\\[2pt]\sin 2\psi&-\cos 2\psi\end{pmatrix},
\label{eq:PolarizationProjector}
\end{equation}
which is independent of the handedness \(\sigma\). In particular,
\(\mathrm{tr}\,\mathsf P=\langle\widehat{\mathbf a}'{}^{\!\top}\widehat{\mathbf a}'\rangle_u=1\)
by our normalization of the homogeneous plane wave [Eqs.~\eqref{Aarbreal} at \(\gamma=0\)].

Averaging Eq.~\eqref{eq:Tuu_instant_general-deriv} over one optical cycle gives:
\begin{equation}
\big\langle T^{(\mathrm{EM})}_{uu}\big\rangle_{u}(x,y)
=\frac{E_0^2}{c^{2}\mu_0}\Big\{\big\langle \mathrm{tr}\,\mathbf D(u)\big\rangle_u
\;+\;2\gamma\,\big\langle\Xi_{\rm H}(u,x,y)\big\rangle_u\Big\}
+{\mathcal O}(\gamma^{2}),
\label{eq:Tuu_avg_general1}
\end{equation}
where \(\Xi_{\rm H}\) is the harmonic-completion scalar defined in Eq.~\eqref{eqXidef-2}.
Since \(\mathrm{tr}\,\mathbf D(u)=1-\cos 2\chi\,\cos(2\theta)\), one has
\(\langle \mathrm{tr}\,\mathbf D\rangle_u=1\), and therefore:
\begin{equation}
\big\langle T^{(\mathrm{EM})}_{uu}\big\rangle_{u}(x,y)
=\frac{E_0^2}{c^{2}\mu_0}\Big\{1
\;+\;2\gamma\,\big\langle\Xi_{\rm H}(u,x,y)\big\rangle_u\Big\}
+{\mathcal O}(\gamma^{2}).
\end{equation}

\begin{remarkx}\label{rem:avg-universality}
In the Maxwell-consistent aligned-null sector, all dependence on the arbitrary transverse envelope
\(f(x,y)\) cancels from \(T_{uu}\) at \(\mathcal O(\gamma)\). Any remaining transverse dependence at this
order is carried only by the harmonic complement \(\psi_{\rm H}\) (equivalently \(\Xi_{\rm H}\)).
Under standard boundary/gauge conventions that exclude nontrivial harmonic complements
(e.g.\ bounded/decaying data on \(\mathbb R^2\), or zero-mode projection on \(\mathbb T^2\)),
one has \(\nabla(\partial_u\psi_{\rm H})\equiv0\) and hence \(\Xi_{\rm H}\equiv0\). In that case
\(\langle T_{uu}\rangle_u\) is spatially homogeneous through \(\mathcal O(\gamma)\).
\end{remarkx}

The cycle-averaged Brinkmann profile \(h_{\rm EM}^{\langle\cdot\rangle}\) satisfies:
\begin{equation}
\Delta_\perp h_{\rm EM}^{\langle\cdot\rangle}(x,y)
=-\frac{16\pi G}{c^4}\,\big\langle T_{uu}^{(\mathrm{EM})}\big\rangle_u(x,y).
\end{equation}
Because the cycle-averaged source has a nonzero spatial mean, any particular solution on \(\mathbb R^2\)
necessarily grows like \(r^2\) (since \(\Delta_\perp r^2=4\)), so one cannot impose
\(h_{\rm EM}^{\langle\cdot\rangle}\to 0\) as \(r\to\infty\). We therefore fix the harmonic ambiguity in \(h\)
by excluding additional vacuum pp--wave modes (harmonic functions/polynomials) and by dropping \(u\)-only gauge
pieces, which carry no curvature.

To display the \(\mathcal O(\gamma)\) structure in the most general (harmonic-complement) case, introduce a
potential \(\Phi_{\Xi}\) solving:
\begin{equation}
\Delta_\perp \Phi_{\Xi}(x,y)=4\,\big\langle\Xi_{\rm H}(u,x,y)\big\rangle_u.
\label{eq:PhiXi-def}
\end{equation}
A convenient particular solution is then:
\begin{equation}
h_{\rm EM}^{\langle\cdot\rangle}(x,y)
=-\mathcal{C}\left\{r^{2}\;+\;2\,\gamma\,\Phi_{\Xi}(x,y)\right\}
+\mathcal O(\gamma^{2}).
\label{eq:hEM_avg_general}
\end{equation}

\begin{remarkx}\label{rem:havg-universal}
Under the standard harmonic-complement removal of Remark~\ref{rem:avg-universality} one has
\(\Phi_{\Xi}\equiv0\), and Eq.~\eqref{eq:hEM_avg_general} reduces to the universal, homogeneous profile
\(
h_{\rm EM}^{\langle\cdot\rangle}(x,y)=-\mathcal C\,r^{2}+\mathcal O(\gamma^{2}),
\)
independent of the choice of transverse modulation \(f(x,y)\) at \(\mathcal O(\gamma)\).
\end{remarkx}

\subsection{Time-dependent correction}\label{TDC}
We now isolate the purely oscillatory (AC) part of the source by subtracting
the cycle average. Define:
\[
T^{(\mathrm{EM})\,\mathrm{osc}}_{uu}
:= T^{(\mathrm{EM})}_{uu}-\big\langle T^{(\mathrm{EM})}_{uu}\big\rangle_u .
\]
In the Maxwell-consistent completion emphasized in Sec.~\ref{Tuusource},
all \emph{profile-driven} \(\mathcal O(\gamma)\) contributions cancel in \(T_{uu}\):
the stress--energy takes the universal form Eq. \eqref{eq:Tuu_instant_general-deriv}, where Eq. \eqref{eqXidef-2} collects only the possible \emph{harmonic-complement} free data
arising in the Coulomb representative (and vanishes under the standard convention
that excludes nontrivial harmonic complements; cf.\ Remark~\ref{rem:avg-universality}).

Using \(\mathrm{tr}\,\mathbf D(u)=1-\cos 2\chi\,\cos(2\theta)\), the oscillatory part is therefore:
\begin{equation}
\label{eq:Xi-osc-def}
\Xi^{\rm osc}_{\rm H}(u,x,y):=\Xi_{\rm H}(u,x,y)-\big\langle\Xi_{\rm H}(u,x,y)\big\rangle_u,
\end{equation}
and
\begin{equation}
\label{eq:Tuu_osc_general}
T^{(\mathrm{EM})\,\mathrm{osc}}_{uu}(u,x,y)
=\frac{E_0^2}{c^{2}\mu_0}\Big\{-\cos 2\chi\,\cos(2\theta)
+2\gamma\,\Xi^{\rm osc}_{\rm H}(u,x,y)\Big\}
+\mathcal O(\gamma^{2}).
\end{equation}
Thus, to \(\mathcal O(\gamma)\) the only unavoidable oscillatory source is the
homogeneous "breathing" term \(\propto \cos 2\chi\,\cos(2\theta)\), which vanishes for
circular polarization (\(\cos 2\chi=0\)). Any additional transverse structure at this order
is encoded solely in the harmonic-complement piece \(\Xi^{\rm osc}_{\rm H}\).

The oscillatory metric correction \(h_{\rm EM}^{\rm osc}\) satisfies:
\begin{equation}
\Delta_\perp h_{\rm EM}^{\rm osc}(u,x,y)
=-\frac{16\pi G}{c^4}\,T^{(\mathrm{EM})\,\mathrm{osc}}_{uu}(u,x,y).
\end{equation}
Introduce \(\Phi_{\Xi}^{\rm osc}\) by:
\begin{equation}\label{eq:PhiXi-osc-def}
\Delta_\perp \Phi_{\Xi}^{\rm osc}(u,x,y)=4\,\Xi^{\rm osc}_{\rm H}(u,x,y),
\end{equation}
so that \(\Delta_\perp r^2=4\).
A convenient particular solution is then:
\begin{equation}
\label{eq:hEM_osc_general}
h_{\rm EM}^{\rm osc}(u,x,y)
=\mathcal{C}\,\cos 2\chi\,\cos(2\theta)\,r^{2}
\;-\;2\mathcal{C}\,\gamma\,\Phi_{\Xi}^{\rm osc}(u,x,y)
\;+\;\mathcal O(\gamma^{2}).
\end{equation}
By construction \(\langle h_{\rm EM}^{\rm osc}\rangle_u=0\), so the full EM contribution to
the Brinkmann profile decomposes as
\(
h_{\rm EM}=h_{\rm EM}^{\langle\cdot\rangle}+h_{\rm EM}^{\rm osc}.
\)

\begin{remarkx}\label{rem:osc-universality}
Under the standard harmonic-complement removal (so \(\Xi_{\rm H}\equiv0\) and hence
\(\Phi_{\Xi}^{\rm osc}\equiv0\)), Eq.~\eqref{eq:Tuu_osc_general} reduces to
\(
T_{uu}^{(\mathrm{EM})\,\mathrm{osc}}=-(E_0^2/c^{2}\mu_0)\,\cos 2\chi\,\cos(2\theta)+\mathcal O(\gamma^2),
\)
and Eq.~\eqref{eq:hEM_osc_general} becomes the universal isotropic breathing profile
\(
h_{\rm EM}^{\rm osc}=\mathcal C\,\cos 2\chi\,\cos(2\theta)\,r^2+\mathcal O(\gamma^2),
\)
independent of the transverse modulation \(f(x,y)\) through \(\mathcal O(\gamma)\).
For circular polarization (\(\cos 2\chi=0\)) the oscillatory sector vanishes to this order.
\end{remarkx}

\subsection{Curvature content and near-axis structure}\label{Curvecont}
In Brinkmann form the only independent Riemann components are:
\begin{equation}
R_{uiuj}(u,x,y)=-\tfrac12\,\partial_i\partial_j h(u,x,y),
\qquad i,j\in\{x,y\}.
\end{equation}
All tidal information is therefore encoded in the transverse Hessian of the
Brinkmann profile, and the local structure of the curvature near a reference
geodesic (e.g.\ the axis $x=y=0$) is controlled by the Taylor expansion of
$h(u,x,y)$ in $(x,y)$.

In the Maxwell-consistent Coulomb representative of Theorem~\ref{thm:gauge-completion},
with the harmonic complement removed as in Remark~\ref{rem:Xi-coulomb-cancel}
(and adopting the same convention as in Secs.~\ref{Cycleavrgprof}–\ref{TDC} of
excluding additional vacuum pp--wave modes),
the electromagnetic stress--energy Eq. \eqref{TuutrD} is spatially homogeneous through $\mathcal O(\gamma)$.

Substituting Eq. \eqref{hEMsolution} into $R_{uiuj}=-\tfrac12\partial_i\partial_j h$ and using
$\partial_i\partial_j r^2=2\delta_{ij}$ gives the universal curvature tensor:
\begin{equation}\label{eq:Riemann_universal_full}
R_{uiuj}(u,x,y)
=\mathcal C\,\mathrm{tr}\,\mathbf D(u)\,\delta_{ij}
+\mathcal O(\gamma^2).
\end{equation}
In particular, \emph{to first order in the modulation parameter $\gamma$ the curvature is independent of the
transverse profile $f(x,y)$} (and in fact independent of $(x,y)$ altogether, once vacuum additions are excluded). The corresponding Petrov classification is discussed in Appendix~\ref{app:geom-class}. In particular, the universal electrovac profile is conformally flat (type O), while any independently superposed vacuum pp--wave yields type N.

\begin{proposition}\label{prop:avg-curvature}
In the Maxwell-consistent completion, the cycle-averaged curvature is:
\begin{equation}\label{eq:Riemann_static_general}
R_{uiuj}^{\langle\cdot\rangle}(x,y)
:=\big\langle R_{uiuj}\big\rangle_u
=\mathcal{C}\,\delta_{ij}
+\mathcal O(\gamma^{2}),
\end{equation}
independent of the polarization parameters $(\psi,\chi,\sigma)$ and of the transverse modulation $f(x,y)$
[at $\mathcal O(\gamma)$].
\end{proposition}

\begin{proof}
From Eq.~\eqref{eq:Riemann_universal_full} one has
$R_{uiuj}=\mathcal C\,\mathrm{tr}\,\mathbf D(u)\,\delta_{ij}+\mathcal O(\gamma^2)$.
Taking the cycle average and using $\langle \mathrm{tr}\,\mathbf D\rangle_u=1$
gives Eq.~\eqref{eq:Riemann_static_general}.
\end{proof}

\medskip
We now analyse the purely time-dependent corrections. Defining the oscillatory part by
$R_{uiuj}^{\rm osc}:=R_{uiuj}-R_{uiuj}^{\langle\cdot\rangle}$ and using
$\mathrm{tr}\,\mathbf D(u)=1-\cos 2\chi\,\cos(2\theta)$, Eq.~\eqref{eq:Riemann_universal_full}
implies:
\begin{equation}\label{eq:Riemann-osc-universal}
R_{uiuj}^{\rm osc}(u,x,y)
= -\,\mathcal C\,\cos 2\chi\,\cos\!\big[2\theta(u)\big]\,\delta_{ij}
+\mathcal O(\gamma^2).
\end{equation}

\begin{proposition}\label{prop:osc-curvature}
In the Maxwell-consistent completion:
\begin{enumerate}
\item For generic (non-circular) polarization, $\cos 2\chi\neq 0$, the oscillatory curvature is:
\begin{equation}
R_{uiuj}^{\rm osc}(u,0)
= -\mathcal{C}\,\cos 2\chi\,
\cos\!\left[2\theta(u)\right]\,\delta_{ij}
+\mathcal O(\gamma^2),
\end{equation}
which is isotropic in the transverse plane, oscillates with phase $2\theta(u)$
(and hence at angular frequency $2\omega$ with respect to $u$), and is independent of the handedness~$\sigma$
to this order.
\item For circular polarization, $\cos 2\chi=0$, one has
\(
R_{uiuj}^{\rm osc}(u,0)=0
\)
to $\mathcal O(\gamma)$ [indeed to $\mathcal O(\gamma)$ everywhere in $(x,y)$ once vacuum modes are excluded].
\end{enumerate}
\end{proposition}

\begin{proof}
This is an immediate specialization of Eq.~\eqref{eq:Riemann-osc-universal} to $x=y=0$
(and the observation that the prefactor vanishes when $\cos 2\chi=0$).
\end{proof}

To connect with local measurements, project the curvature onto an orthonormal tetrad
\(e_{(0)}^\mu,e_{(i)}^\mu\) comoving with a timelike reference observer.
Along a chosen worldline one defines:
\begin{equation}
R_{0i0j}(T,0)
:= e_{(0)}^{\ \mu} e_{(i)}^{\ \nu} e_{(0)}^{\ \rho} e_{(j)}^{\ \sigma}
   R_{\mu\nu\rho\sigma}\big|_{\text{worldline}},
\end{equation}
where \(T\) is proper time and \(u=u(T)\). Since Brinkmann pp--waves have nonzero curvature
only in the \(uiuj\) components, one has along the worldline:
\begin{equation}
R_{0i0j}(T,0)=\big[\dot u(T)\big]^{2}\,R_{uiuj}\big[u(T),0\big],
\qquad \dot u:=\frac{\mathrm d u}{\mathrm d T}.
\end{equation}
Thus the tidal matrix measured by the observer has the same transverse tensor structure
as \(R_{uiuj}\), up to the positive scalar factor \(\dot u^{\,2}\). In particular,
Proposition~\ref{prop:osc-curvature} implies:
\begin{equation}
\label{eq:R0i0j-osc-general}
R_{0i0j}^{\rm osc}(T,0)
= -\big[\dot u(T)\big]^{2}\,\mathcal{C}\,\cos 2\chi\,
\cos\!\left\{2\theta[u(T)]\right\}\,\delta_{ij}
+\mathcal O(\gamma^2),
\end{equation}
for non-circular polarization, and \(R_{0i0j}^{\rm osc}(T,0)=0\) for circular polarization at this order.

\begin{remarkx}
The separation vector $X^i(T)$ between nearby geodesics in a local inertial frame satisfies the
geodesic deviation equation \(\ddot X^i(T) = - R_{0i0j}(T,0)\,X^j(T)\). Therefore
\(R_{0i0j}\) acts as a time-dependent \emph{tidal operator} governing local focusing/defocusing and oscillatory
distortions. In the Maxwell-consistent electrovac Brinkmann family, the tidal field
through $\mathcal O(\gamma)$ consists of a universal, polarization-independent cycle-averaged part plus an
isotropic oscillatory component present only for non-circular polarization. Any transverse structure at the
same perturbative order would have to come from an \emph{independently chosen} vacuum pp--wave addition
(\(\Delta_\perp h_{\rm vac}=0\)), which we exclude by convention when isolating the electrovac response.
\end{remarkx}

The electrovac Brinkmann solutions constructed here exhibit a universal curvature response at $\mathcal O(\gamma)$: the weak transverse modulation $f(x,y)$ does not alter the
Brinkmann curvature at this order, and the on-axis (indeed global) tidal field decomposes into a universal
cycle-averaged isotropic part and a polarization-controlled isotropic oscillatory part.

\subsection{Superposition with a vacuum gravitational wave}\label{GWsuperpos}
For Brinkmann pp--waves:
\begin{equation}\label{RuuGW}
R_{uu}=-\tfrac12\,\Delta_{\perp}h,\qquad
R_{uiuj}=-\tfrac12\,\partial_i\partial_j h \quad (i,j\in\{x,y\}),
\end{equation}
so adding any term depending \emph{only} on $u$ leaves $\Delta_{\perp}h$ unchanged and thus does not affect the field equations. Such an addition is pure gauge: the coordinate shift:
\begin{equation}\label{vmapstov}
v\;\mapsto\; v-\tfrac12\!\int^{u}\!\mathfrak f(\tilde u)\,\mathrm d\tilde u
\end{equation}
removes $\,\mathfrak f(u)$ from $h$, and produces no curvature because $R_{uiuj}$ depends only on transverse derivatives of $h$. Hence a genuine gravitational contribution requires nontrivial $(x,y)$--dependence (so that $\partial_i\partial_j h\neq0$); the simplest vacuum choice is quadratic and harmonic in $(x,y)$ with $u$--dependent amplitudes.

To model a plane vacuum GW aligned with the same null direction, add the standard Brinkmann vacuum (pp--wave) profile:
\begin{equation}\label{hG_quadratic}
h_{\mathrm{GW}}(u,x,y)
=\tfrac12\,H_{+}(u)\,(x^{2}-y^{2})
+H_{\times}(u)\,x y,
\end{equation}
for which $\Delta_{\perp}h_{\mathrm{GW}}=0$ and therefore $R_{uu}^{\mathrm{GW}}=0$ while:
\begin{equation}\label{eq:Ruiuj-GW-matrix}
R_{uiuj}^{\mathrm{GW}}
=-\frac12
\begin{pmatrix}
\partial_x^2 h_{\mathrm{GW}} & \partial_x\partial_y h_{\mathrm{GW}}\\[2pt]
\partial_y\partial_x h_{\mathrm{GW}} & \partial_y^2 h_{\mathrm{GW}}
\end{pmatrix}
=-\frac12
\begin{pmatrix}
H_{+}(u) & H_{\times}(u)\\[2pt]
H_{\times}(u) & -\,H_{+}(u)
\end{pmatrix}.
\end{equation}
For completeness, the Petrov type of the vacuum pp--wave sector and its superposition with the electrovac background are summarized in Appendix~\ref{app:geom-class}.

A convenient monochromatic specialization is:
\begin{equation}\label{HuGW}
H_{+}(u)=\mathcal A_{+}\cos\!\Big(\frac{\Omega}{c}\,u+\phi_{+}\Big),\qquad
H_{\times}(u)=\mathcal A_{\times}\cos\!\Big(\frac{\Omega}{c}\,u+\phi_{\times}\Big),
\end{equation}
with $\Omega$ the GW angular frequency conjugate to the Brinkmann coordinate $u$
(\(\Omega=\sqrt2\,\Omega_{\rm phys}\) for an inertial observer’s frequency \(\Omega_{\rm phys}\)). Right/left circular polarization corresponds to equal amplitudes and a $\pm\pi/2$ phase shift, e.g.
$H_{\times}(u)=\pm H_{+}\!\big(u+\tfrac{\pi c}{2\Omega}\big)$.

Because $R_{uu}=-\tfrac12\Delta_{\perp}h$ is \emph{linear} in $h$ for the Brinkmann ansatz, the sum:
\begin{equation}
\label{eq:h_total_general}
h(u,x,y)
=h_{\rm EM}^{\langle\cdot\rangle}(x,y)\;+\;h_{\rm EM}^{\rm osc}(u,x,y)\;+\;h_{\mathrm{GW}}(u,x,y)
\end{equation}
solves the Einstein--Maxwell system with the same electromagnetic source:
\begin{equation}
\Delta_{\perp}\Big[h_{\rm EM}^{\langle\cdot\rangle}(x,y)+h_{\rm EM}^{\rm osc}(u,x,y)\Big]
=-\,\frac{16\pi G}{c^{4}}\,T_{uu}^{(\mathrm{EM})}(u,x,y),
\qquad
\Delta_{\perp}h_{\mathrm{GW}}(u,x,y)=0.
\end{equation}
The total tidal field is obtained by differentiating Eq. \eqref{eq:h_total_general}:
\begin{equation}\label{Ruiuj_sum}
R_{uiuj}=-\tfrac12\,\partial_i\partial_j h
=R_{uiuj}^{\rm EM}\;+\;R_{uiuj}^{\mathrm{GW}},
\end{equation}
with $R_{uiuj}^{\rm EM}$ computed from the electrovac part and $R_{uiuj}^{\mathrm{GW}}$ given explicitly by Eq. \eqref{eq:Ruiuj-GW-matrix}.
\begin{remarkx}
The statements above hold for \emph{arbitrary} smooth functions $H_{+}(u),H_{\times}(u)$: any choice yields a vacuum pp--wave because $\Delta_{\perp}h_{\mathrm{GW}}=0$. The units are consistent with $[H_{\{+,\times\}}]=\text{length}^{-2}$ so that $h_{\mathrm{GW}}$ is dimensionless.
\end{remarkx}

\section{Results}\label{Results}
\subsection{Special transverse profiles \(f(x,y)\): normal forms and a recipe}\label{subsec:why-special-cases}

The purpose of introducing explicit transverse profiles \(f(x,y)\)
is \emph{not} to generate new \(\mathcal O(\gamma)\) gravitational structure in \(T_{uu}\) or in the
curvature (Sec.~\ref{Curvecont}), but rather:
(i) to construct explicit gauge--completed Maxwell potentials \(A_\mu\) realizing a prescribed weak
transverse modulation at the level of the vector potential, and
(ii) to provide analytically tractable testbeds [useful e.g.\ for higher--order \(\mathcal O(\gamma^2)\)
calculations, for off--axis field diagnostics, or for optional superposition with vacuum pp--waves].

At an intermediate algebraic level (before enforcing the cancellation in
Remark~\ref{rem:Xi-coulomb-cancel}) one may organize the \(\mathcal O(\gamma)\) correction to the
instantaneous flux in terms of three scalar contractions built from \(f\) and polarization:
an \emph{isotropic} modulation \(f\) itself, an \emph{anisotropic} projection\footnotemark[2]
\(S{:}D\) (with \(S=\nabla\nabla w\) in the canonical completion and \(D=\widehat{\mathbf a}'\widehat{\mathbf a}'{}^{\!\top}\)),
and a \emph{completion} scalar \(\Xi\) encoding the \(u\)-dependence of the Coulomb--gauge scalar
\(\psi(u,x,y)\) (equivalently, of the transverse gradient correction \(\gamma\,\nabla\psi\) in \(A_i\))
that is required by source--free Maxwell equations when \(f\neq0\) (Sec.~\ref{subsec:arbPolEMwaveweakf} and Sec.~\ref{Tuusource}).
In the explicit Coulomb representative, and under boundary conditions that remove the harmonic complement,
these \(\mathcal O(\gamma)\) terms cancel (up to removable harmonic data), yielding the universal result
Eq.~\eqref{TuutrD}.

\footnotetext[2]{We use the Frobenius (double‑contraction) product on $2\times2$ tensors:
$A{:}B := \mathrm{tr}(A^{\mathrm T}B) = \sum_{i,j} A_{ij} B_{ij}$. In particular,
$S{:}D=\sum_{i,j} S_{ij} D_{ij}$. If $D=\mathbf p\,\mathbf p^{\mathrm T}$ for a unit vector
$\mathbf p$, then $S{:}D=\mathbf p^{\mathrm T} S\,\mathbf p$. The scalar $A{:}B$ is rotation‑invariant:
$(RAR^{\mathrm T}){:}(RBR^{\mathrm T})=A{:}B$ for any orthogonal $R$.}

Given a smooth transverse profile \(f\) with \(0<|\gamma|\ll1\), a Maxwell--consistent electrovac solution
[to \(\mathcal O(\gamma)\) in the Maxwell sector and yielding the universal gravitational response
Eq. \eqref{hEMsolution}] can be constructed as follows:
\begin{enumerate}
\item \emph{(Choose the carrier wave)} Fix the homogeneous plane--wave functions \(a_x(u),a_y(u)\) [or equivalently
\(\mathbf a(u)\)] and define \(\mathbf D(u)=\widehat{\mathbf a}'(u)\widehat{\mathbf a}'(u)^{\!\top}\) as in
Sec.~\ref{Tuusource}. This fixes \(\mathrm{tr}\,\mathbf D(u)\) and hence the universal electrovac source
Eq. \eqref{TuutrD}.

\item \emph{(Canonical longitudinal completion)} Solve the scalar Poisson equation
\(\Delta_\perp w=f\) and set \(g=\nabla w\) (Remark~\ref{rem1}), so that
\(S=\nabla\nabla w\) and \(\mathrm{tr}\,S=f\).
Then define:
\begin{equation}
A_u(u,x,y)=\gamma\,\mathbf a'(u)\!\cdot\!\nabla w(x,y),
\qquad A_v=0.
\end{equation}

\item \emph{(Coulomb transverse completion by a single scalar)} For each fixed \(u\), solve the scalar Poisson equation:
\begin{equation}
\Delta_\perp\psi(u,x,y)=\mathbf a(u)\!\cdot\!\nabla f(x,y),
\end{equation}
with boundary/normalization conditions chosen so that the transverse harmonic complement is removed
(e.g.\ decay on \(\mathbb R^2\), or zero--mean on \(\mathbb T^2\)).
Then set:
\begin{equation}
A_i(u,x,y)=a_i(u)+\gamma\,\partial_i\psi(u,x,y),\qquad i=x,y,
\end{equation}
which is exactly the explicit representative in Theorem~\ref{thm:gauge-completion}.

\item \emph{(Maxwell consistency and universality of the source)} With the harmonic complement removed,
\(F_{ui}=a'_i(u)+\mathcal O(\gamma^2)\) and therefore we obtain Eq.~\eqref{TuutrD},
independent of \(f\) to \(\mathcal O(\gamma)\) (Remark~\ref{rem:Xi-coulomb-cancel}).

\item \emph{(Assemble \(h\))} Solve \(\Delta_\perp h=-(16\pi G/c^4)\,T_{uu}^{(\mathrm{EM})}\) and take the universal
particular solution Eq. \eqref{hEMsolution}. Any additional solution of \(\Delta_\perp h_{\rm vac}=0\) is a
vacuum pp--wave mode and may be included or excluded independently (Sec.~\ref{GWsuperpos}).
\end{enumerate}

\begin{lemma}
\label{lem:degree-two-normal-forms}
Every centered harmonic quadratic can be written as
\(
f(x,y)=a\,\frac{x^2-y^2}{2}+b\,x y
\)
with $a,b\in\mathbb R$.
Under a transverse rotation by angle $\varphi$,
\(
(a,b)\mapsto (a',b')=(a\cos2\varphi+b\sin2\varphi,\,-a\sin2\varphi+b\cos2\varphi).
\)
Thus any such $f$ is SO(2)‑equivalent to either
\(\tfrac12(x^2-y^2)\) or \(xy\).
\end{lemma}
\begin{proof} This is the standard spin‑2 action on the pair \((a,b)\). \end{proof}

\begin{remarkx}
\label{rem:two-cases-suffice}
By Lemma~\ref{lem:degree-two-normal-forms}, the two harmonic quadrupoles
\(f_1=\tfrac12(x^2-y^2)\) and \(f_2=xy\)
form canonical representatives of all centered degree‑two inhomogeneities up to rotation.
Hence, changing \(f\) changes the explicit completion fields \(w,\psi\) (and hence \(A_u\) and the
pure‑gradient transverse correction), but it does \emph{not} change the sourced electrovac curvature to \(\mathcal O(\gamma)\).
Relative angles between a quadrupolar \(f\) and the polarization parameters enter only through the Maxwell potential data,
or through optional, independently superposed vacuum pp--wave modes.
\end{remarkx}

\begin{remarkx}
\label{rem:iso-aniso}
Within the intermediate \(\mathcal O(\gamma)\) expansion of Sec.~\ref{Tuusource}, one can organize the would‑be profile
dependence schematically as:
\(
T_{uu}\ \propto\ \operatorname{tr}\mathbf D\,[1+2\gamma f]\;-\;2\gamma\,S{:}\mathbf D\;+\;2\gamma\,\Xi,
\)
where \(S{:}\mathbf D\) encodes anisotropic projections and \(\Xi\) represents the \(u\)-dependent completion contribution.
In the explicit Coulomb representative with harmonic complement removed,
\(\Xi\) cancels the \(f\) and \(S{:}\mathbf D\) contributions at \(\mathcal O(\gamma)\),
leaving the universal, spatially homogeneous source Eq. \eqref{TuutrD}.
Thus any apparent \(\mathcal O(\gamma)\) \(f\)-dependence in the Brinkmann profile can only arise from
vacuum (harmonic) pp--wave data, which is independent of the electromagnetic source and may be excluded by convention.
\end{remarkx}

\begin{remarkx}
\label{rem:scaling-locality}
If \(f\) is homogeneous of degree \(n\), then the canonical completion potential \(w\) solving \(\Delta_\perp w=f\)
is homogeneous of degree \(n+2\), and the scalar \(\psi\) solving \(\Delta_\perp\psi=\mathbf a\!\cdot\!\nabla f\)
inherits the same degree shift. In particular, for degree‑two (quadrupolar) \(f\), the completion fields contribute
only quartic (and higher) transverse structure in \(w\) and \(\psi\), consistent with the near‑axis universality of the
curvature in Sec.~\ref{Curvecont}. Any genuine modification of the \emph{sourced} curvature by \(f\) would therefore occur
only beyond the present order [e.g.\ at \(\mathcal O(\gamma^2)\)], or by adding an independent vacuum pp--wave sector.
\end{remarkx}

As a quick demonstration of the above statements, let us consider the function $f(x,y)=xy$, which is a centered harmonic quadrupole (\(\Delta_\perp f=0\)), with
\(\partial_x f=y\) and \(\partial_y f=x\).
In the canonical (gradient) completion one may choose an explicit polynomial solution of
\(\Delta_\perp w=f\), for instance:
\begin{equation}\label{eq:w_xy}
w(x,y)=\frac{1}{12}\,r^{2}\,x y=\frac{1}{12}\big(x^{3}y+x y^{3}\big),
\qquad r^{2}=x^{2}+y^{2},
\end{equation}
which is regular at the origin and satisfies \(\Delta_\perp w=xy\).
Then \(g=\nabla w\) and \(A_u=\gamma\,[a'_x(u)\,\partial_x w+a'_y(u)\,\partial_y w]\).

For the Maxwell-consistent Coulomb representative one takes:
\(
A_i(u,x,y)=a_i(u)+\gamma\,\partial_i\psi(u,x,y),
\)
with \(\psi\) solving the scalar Poisson equation
\(\Delta_\perp\psi=\mathbf a(u)\!\cdot\!\nabla f=a_x(u)\,y+a_y(u)\,x\).
A convenient explicit polynomial choice is:
\begin{equation}\label{eq:psi_xy}
\psi(u,x,y)=\frac{r^{2}}{8}\,\big[a_y(u)\,x+a_x(u)\,y\big],
\qquad \Delta_\perp\psi=a_y(u)\,x+a_x(u)\,y.
\end{equation}
With this completion one has \(F_{xy}=0\) identically, and the
\(\mathcal O(\gamma)\) profile-dependent contributions to \(T^{(\mathrm{EM})}_{uu}\) cancel
(up to the removable harmonic complement fixed to zero by the chosen boundary/gauge convention).
Therefore, to first order in \(\gamma\), the Brinkmann profile is the same as for the
homogeneous electrovac plane wave.

Accordingly, the cycle-averaged and oscillatory pieces take the universal form:
\begin{align}
h_{\rm apEM}^{\langle\cdot\rangle}(x,y)
&=-\mathcal{C}\,r^{2}+\mathcal O(\gamma^{2}),
\label{eq:hapEM_avg_xy}\\
h_{\rm apEM}^{\rm osc}(u,x,y)
&=\mathcal{C}\;\cos 2\chi\;\cos\!\left[2\theta(u)\right]\;r^{2}
+\mathcal O(\gamma^{2}),
\label{eq:hapEM_osc}
\end{align}
in agreement with the cancellation statement of Remark~\ref{rem:Xi-coulomb-cancel}.

For a circularly polarized electromagnetic wave (\(\cos 2\chi=0\)) the homogeneous oscillatory term vanishes,
and the same Maxwell-consistent completion gives, to \(\mathcal O(\gamma)\),
\(
h_{\rm cpEM}^{\langle\cdot\rangle}(x,y)=-\mathcal{C}\,r^{2}
+\mathcal O(\gamma^{2}),\
h_{\rm cpEM}^{\rm osc}(u,x,y)=\mathcal O(\gamma^{2}).
\)

\begin{remarkx}\label{rem:counterexample-not-electrovac}
The cancellation/universality result of the electrovac, aligned--null sector hinges on enforcing
the \emph{source--free} Maxwell equations $\nabla_\nu F^{\mu\nu}=0$ together with $A_v=0$, $\partial_v A_\mu=0$ and
$F_{xy}=0$.
If one instead imposes a transverse envelope by hand \emph{without} enforcing the source--free Maxwell equations,
the modulation survives in $T_{uu}$ at $\mathcal O(\gamma)$ but is supported by a nonzero current $J^\mu\neq 0$ (hence the configuration is \emph{not} electrovac).

For example, take a linearly polarized potential with $A_v=A_u=0$ and:
\begin{equation}
A_x(u,x,y)=a(u)\,\big[1+\gamma f(x)\big],\qquad A_y(u,x,y)=0,
\qquad 0<|\gamma|\ll1,
\end{equation}
where $f$ is any nonconstant smooth function of $x$ alone (so that $F_{xy}=0$ remains true).
Then:
\begin{equation}
F_{ux}=\partial_u A_x=a'(u)\,\big[1+\gamma f(x)\big],\qquad F_{uy}=0,
\end{equation}
and hence the electromagnetic energy flux becomes:
\begin{equation}
T^{(\mathrm{EM})}_{uu}(u,x,y)
=\frac{1}{\mu_0}\big(F_{ux}^2+F_{uy}^2\big)
=\frac{1}{\mu_0}\,a'(u)^2\,\big[1+2\gamma f(x)\big]+\mathcal O(\gamma^2),
\end{equation}
which is explicitly $x$--dependent at $\mathcal O(\gamma)$.

However, the $v$--component Maxwell equation fails: since $\sqrt{-g}=1$ in Brinkmann coordinates and
$F^{vi}=-F_{ui}$ for the aligned ansatz, the sourced Maxwell system reads
$\nabla_\nu F^{\mu\nu}=\mu_0 J^\mu$ and gives:
\begin{equation}
\mu_0 J^{v}=\partial_i F^{vi}=-\partial_i F_{ui}=-\partial_x F_{ux}
=-\gamma\,a'(u)\,f'(x)\,,
\end{equation}
so that:
\begin{equation}
J^{v}(u,x)=-\frac{\gamma}{\mu_0}\,a'(u)\,f'(x)\neq 0
\qquad\text{(unless $f$ is constant).}
\end{equation}
All other components $J^u,J^x,J^y$ vanish for this simple example.
Thus the transverse modulation is supported by a \emph{null current} flowing along the pp--wave direction
$k=\partial_v$, and the configuration lies outside the \emph{electrovac} class treated in this paper.

More generally, for the naive ansatz $A_i=a_i(u)\,[1+\gamma f(x,y)]$ with $A_u=0$ one typically also generates
$F_{xy}=\gamma\,(a_y\,\partial_x f-a_x\,\partial_y f)\neq 0$, so the field ceases to be aligned--null as well.
Either way, an $\mathcal O(\gamma)$ transverse envelope that survives in $T_{uu}$ requires leaving the strict
source--free aligned--null electrovac pp--wave sector.
\end{remarkx}

\subsection{Test--particle dynamics in the Brinkmann pp--wave}\label{subsec:test-particle-dynamics}

This subsection is included purely as a \emph{pedagogical} visualization of the universal
near-axis tidal content established in Sec.~\ref{Curvecont}. In the electrovac
setting (aligned null Maxwell field with source-free Maxwell equations enforced and with
the harmonic complement removed), a weak transverse "envelope" $1+\gamma f(x,y)$ does
\emph{not} generate any quadrupolar/shearing contribution to the sourced Brinkmann profile
at $\mathcal O(\gamma)$. Accordingly, to the order considered here the electromagnetic
contribution is isotropic in the transverse plane, and a circular ring remains a circle
under the electromagnetic tidal field (no shear), up to the independent addition of an
optional vacuum pp--wave sector.

We illustrate this by considering a circularly polarized electromagnetic wave (cpEM),
together with a right--circularly polarized co--propagating vacuum gravitational pp--wave.
Introducing the fundamental quantities \(\alpha=e^{2}/(4\pi\varepsilon_{0}\hbar c)\)
(fine‑structure constant), \(m_{\mathrm{Pl}}=\sqrt{\hbar c/G}\) (Planck mass),
\(\overline{\lambda}_{C}=\hbar/(m_{e}c)\) (electron reduced Compton wavelength), and
\(E_{S}=m_{e}c^{2}/(e\,\overline{\lambda}_{C})\) (Schwinger critical field), we write a
convenient dimensionless Brinkmann profile as:
\begin{equation}\label{Brinkmann_metric_circwave}
h(u,x,y)= -\frac{m_{e}^{2}}{\alpha\,m_{\mathrm{Pl}}^{2}}\;\mathcal{A}_{E}\;
\left(x^{2}+y^{2}\right)+\;\Re\left\{\frac{\mathcal{A}_{G}(\Omega)}{2}(x-iy)^2e^{i\Omega u/c}\right\},
\end{equation}
where \(\mathcal{A}_{E}\equiv(E_{0}/E_{S})^{2}\); for simplicity we take
\(\phi_{0}=\phi_{\times}=\phi_{+}=0\). The GW contribution is written in the proper
Brinkmann (quadrupolar, harmonic) form; a term depending only on \(u\) would be pure
gauge and carry no curvature.

For a monochromatic TT--gauge strain \(h_{+,\times}(u)\propto \cos(\Omega u/c)\), the
corresponding Brinkmann tidal coefficients satisfy
\(H_{+,\times}(u)=-(\Omega^{2}/c^{2})\,h_{+,\times}(u)\).
Therefore, for the complex right--circular polarization amplitude in
Eq.~\eqref{Brinkmann_metric_circwave} one may write:
\begin{equation}\label{AGstrain}
\mathcal{A}_{G}(\Omega)=\left(\frac{\Omega\overline{\lambda}_{C}}{c}\right)^{2}\sqrt{2}\,h_{0},
\end{equation}
where \(h_{0}\) is the RMS strain per polarization.

The test particles follow Brinkmann geodesics in the transverse plane,
\begin{equation}
\frac{{\rm d}^2 X_i}{{\rm d}u^2}
=-\frac{1}{2}\,\partial_i h[u,X(u)],
\qquad i=x,y,
\label{eq:Brinkmann-geodesic-num}
\end{equation}
which can be integrated numerically starting from an uniformly sampled ring of
radius $R_0$ with zero initial velocity.

For visualisation one can adopt a simple first--order update in the affine
parameter $u$. Writing $\dot{\boldsymbol X}=\boldsymbol V$ and
$\boldsymbol a(u,\boldsymbol X)=-\tfrac12\nabla_\perp h$, the update
from $u_n$ to $u_{n+1}=u_n+\Delta u$ is
\begin{equation}
\boldsymbol V_{n+1}=\boldsymbol V_n+\Delta u\,\boldsymbol a(u_n,\boldsymbol X_n),
\qquad
\boldsymbol X_{n+1}=\boldsymbol X_n+\Delta u\,\boldsymbol V_{n+1}.
\label{eq:update-scheme}
\end{equation}
This basic scheme is sufficient for short-time qualitative snapshots. It is not
intended for high-accuracy long-time evolution in regimes with strong defocusing.

\subsubsection{Electromagnetic component}

In the electrovac setting, a circularly polarized electromagnetic wave produces,
through $\mathcal O(\gamma)$, a \emph{purely isotropic} Brinkmann contribution:
\begin{equation}
h_{\rm EM}(u,x,y)=
-\,\frac{m_{e}^{2}}{\alpha\,m_{\mathrm{Pl}}^{2}}\;\mathcal{A}_{E}\;r^2
\;+\;\mathcal O(\gamma^{2}),
\qquad r^2=x^2+y^2,
\end{equation}
i.e.\ there is no $\mathcal O(\gamma)$ quadrupolar/shear term sourced by a transverse envelope
$1+\gamma f(x,y)$ once Maxwell consistency is enforced and the harmonic complement is removed.
(Any additional anisotropic quadratic piece would correspond to an \emph{independently added}
vacuum pp--wave mode, not to the electromagnetic source.)

The corresponding transverse equations of motion are, to leading order,
\begin{equation}
\ddot x\simeq \kappa_E\,x,\qquad
\ddot y\simeq \kappa_E\,y,
\qquad
\kappa_E:=\frac{m_{e}^{2}}{\alpha\,m_{\mathrm{Pl}}^{2}}\;\mathcal{A}_{E},
\end{equation}
so the electromagnetic background drives an isotropic defocusing: an initially circular ring
remains circular (no shear) and expands/defocuses radially.
For non-circular electromagnetic polarization one may additionally have an \emph{isotropic}
oscillatory “breathing” factor at frequency $2\omega$ (Sec.~\ref{TDC}), but the transverse
tidal matrix remains proportional to $\delta_{ij}$ to this order.

Restoring physical units, \(\kappa_E\sim m_e^2/(\alpha m_{\mathrm{Pl}}^2)\lesssim 10^{-43}\times(E_0/E_S)^2\), so any such effect is extremely small unless one considers unrealistically
large fields and/or very long affine times. In numerical illustrations one therefore treats
\(\kappa_E\) as an adjustable dimensionless strength parameter to make the tidal pattern visible.

\subsubsection{Gravitational--wave component}

For the GW part of Eq.~\eqref{Brinkmann_metric_circwave} the transverse tidal field is traceless and
purely quadrupolar. Differentiating the explicit right--circular form gives:
\begin{align}
\partial_x h_{\rm GW}
&=\mathcal A_G(\Omega)\big[x\cos(\Omega u/c)+y\sin(\Omega u/c)\big],\\
\partial_y h_{\rm GW}
&=\mathcal A_G(\Omega)\big[-y\cos(\Omega u/c)+x\sin(\Omega u/c)\big],
\end{align}
which enter directly into the geodesic update rule Eq. \eqref{eq:update-scheme}. Starting from a
circular ring of radius $R_0$ and zero initial velocity, the snapshots at different $u$
display the familiar sequence of rotating ellipses associated with circularly polarized
plane gravitational waves: the ring is sheared along orthogonal directions while the principal
axes rotate with the GW phase.

\subsubsection{Combined electromagnetic and gravitational driving}

When both components are present, the particles evolve under the full profile
Eq.~\eqref{Brinkmann_metric_circwave}. To leading order one may separate the acceleration into an
approximately isotropic part and a traceless quadrupolar shear,
\begin{equation}
\ddot{\boldsymbol X}(u)
\simeq \kappa_E\,\boldsymbol X
\;+\;\frac{1}{2}\,\mathcal A_G(\Omega)\,\nabla_\perp\!\left[\Re\!\left\{\frac{(x-iy)^2}{2}e^{i\Omega u/c}\right\}\right]_{\boldsymbol X}.
\end{equation}
Here the electromagnetic contribution produces a secular isotropic defocusing (inflation of the ring),
whereas the GW contribution produces the rotating quadrupolar shear. Crucially, in the electrovac
framework there is \emph{no} additional $\mathcal O(\gamma)$ electromagnetic shear sourced by a transverse
envelope $f(x,y)$: any anisotropy seen in the ring is controlled by the vacuum GW sector (or by other
independently specified vacuum pp--wave modes), while the electromagnetic backreaction remains isotropic
through $\mathcal O(\gamma)$.

\section{Summary and concluding remarks}\label{sec:summary}
We constructed and analyzed a class of aligned Einstein--Maxwell solutions in Brinkmann (pp--wave) form
in which the transverse vector potential is \,\emph{a priori}\, allowed to carry a weak, slowly varying
transverse modulation.
The key point of the paper is that, within the electrovac pp--wave/aligned--null class,\footnote{In the
sense that the only matter field is a source--free Maxwell field and the principal null direction of the
Maxwell field coincides with the covariantly constant Brinkmann vector $k=\partial_v$.}
this apparent freedom is strongly constrained by Maxwell consistency.

\medskip
\noindent\textbf{Maxwell consistency and gauge completion.}
Starting from a transversely modulated ansatz for $A_i$, we showed that taking $A_u=0$ generically violates
$\nabla_\nu F^{\mu\nu}=0$ already at $\mathcal O(\gamma)$.
A minimal, polarization--agnostic completion is obtained by solving the transverse Poisson problems
\(\Delta_\perp w=f\) and \(\Delta_\perp\psi=\mathbf a(u)\!\cdot\!\nabla f+\Delta_\perp\psi_{\rm H}\), and setting
\(A_u=\gamma\,\mathbf a'(u)\!\cdot\!\nabla w\) and \(A_i=a_i(u)+\gamma\,\partial_i\psi\)
(Theorem~\ref{thm:gauge-completion}).
The only residual freedom is the harmonic datum $\psi_{\rm H}$, which corresponds to adding a vacuum
(pp--wave) mode and is fixed by the chosen boundary/gauge convention.

\medskip
\noindent\textbf{Universality of the electrovac source.}
For the Maxwell--consistent Coulomb representative, the electromagnetic flux component takes the universal form:
\begin{equation}
T^{(\mathrm{EM})}_{uu}(u,x,y)=\frac{E_0^2}{c^{2}\mu_0}\,\mathrm{tr}\,D(u)
+\frac{E_0^2}{c^{2}\mu_0}\,\gamma\,\mathbf a'(u)^{\!\top}\nabla\psi_{\rm H}(x,y)
+\mathcal O(\gamma^{2}),
\end{equation}
so that, under decay/zero--mode boundary conditions eliminating $\psi_{\rm H}$, the source is transversely
homogeneous through $\mathcal O(\gamma)$ [Eq.~\eqref{TuutrD}].
In other words: \emph{insisting on electrovac Maxwell equations removes the apparent $\mathcal O(\gamma)$
$(x,y)$--dependence of $T_{uu}$, up to optional harmonic/vacuum data}.

\medskip
\noindent\textbf{Metric profile and curvature universality.}
Since $R_{uu}=-\tfrac12\Delta_\perp h$, the Brinkmann profile sourced by the electromagnetic wave coincides,
through $\mathcal O(\gamma)$, with the standard homogeneous electrovac plane--wave profile,
\(h_{\rm EM}^{\langle\cdot\rangle}=-\mathcal C\,r^2\) plus the universal oscillatory ``breathing'' term
\(h_{\rm EM}^{\rm osc}=\mathcal C\cos2\chi\cos(2\theta)\,r^2\) for noncircular polarization.
Consequently the on--axis tidal matrix is isotropic and independent of the modulation profile $f$:
\(R_{uiuj}(u,0)\propto\delta_{ij}\) with an oscillatory amplitude proportional to $\cos2\chi$.
Any residual $\mathcal O(\gamma)$ transverse structure is harmonic and therefore belongs to the vacuum
pp--wave sector rather than to the electromagnetic source.

\medskip
\noindent\textbf{What happens if Maxwell is not source--free.}
To emphasize the logical necessity of the cancellation, we provided a short ``counterexample''
(Remark~\ref{rem:counterexample-not-electrovac}): keeping the naive modulated ansatz with $A_u=0$ produces an
$\mathcal O(\gamma)$ transverse modulation of $T_{uu}$, but at the price of introducing a nonzero current
$J^\mu\neq0$.
This clarifies the physical meaning of transverse envelopes in the Brinkmann class:
if $T_{uu}$ varies across the wavefront at $\mathcal O(\gamma)$, then the field is not electrovac.

\medskip
\noindent\textbf{Outlook.}
The present analysis isolates a clean universality statement at first order in a weak transverse modulation.
Genuine profile--dependent backreaction can enter at higher order (\(\mathcal O(\gamma^2)\)) or in more general
setups (finite beams, non--null Maxwell fields, gyratons, or additional matter such as a plasma).
These directions provide natural extensions if one wishes to model physically localized electromagnetic beams
beyond the electrovac pp--wave idealization.

\appendix
\section{Geometric classification of the Brinkmann electrovac solution}
\label{app:geom-class}

\subsection{Principal null structure and Kundt character}

In Brinkmann form Eq.~\eqref{Brinkmann_metric}, the vector field
\(
k^\mu\partial_\mu=\partial_v
\)
is null, geodesic, and covariantly constant:
\(
k^\mu k_\mu=0,\;
k^\nu\nabla_\nu k^\mu=0,\;
\nabla_\nu k^\mu=0.
\)
Hence the spacetime belongs to the Kundt class and, more specifically, to the pp--wave subclass.

It is convenient to introduce the transverse Hessian of the Brinkmann profile and its tracefree part:
\begin{equation}
\label{eq:app-Hess-Delta-Tracefree}
\mathcal H_{ij}(u,x,y):=\partial_i\partial_j h,
\
\Delta_\perp h:=\delta^{ij}\mathcal H_{ij}=\partial_x^2 h+\partial_y^2 h,
\
\mathcal S_{ij}:=\mathcal H_{ij}-\tfrac12\delta_{ij}\,\Delta_\perp h,
\ i,j\in\{x,y\}.
\end{equation}
The only nonzero Riemann components are:
\begin{equation}
\label{eq:app-Riemann-Ricci}
R_{uiuj}=-\tfrac12\,\mathcal H_{ij},
\qquad\text{and}\qquad
R_{uu}=-\tfrac12\,\Delta_\perp h,
\end{equation}
so the Weyl tensor reduces to:
\begin{equation}
\label{eq:app-Weyl}
C_{uiuj}=-\tfrac12\,\mathcal S_{ij},
\qquad
\text{all other}\ C_{abcd}=0.
\end{equation}

With the standard Newman--Penrose tetrad aligned with \(k\):\cite{NewmanPenrose1962}
\begin{equation}
k^a=\partial_v,\quad
\ell^a=\partial_u-\tfrac12 h\,\partial_v,\quad
m^a=\tfrac{1}{\sqrt2}\big(\partial_x+i\,\partial_y\big),
\end{equation}
the only nonvanishing Weyl scalar is:
\begin{equation}
\label{eq:app-Psi4}
\Psi_4 \;=\; -\,C_{abcd}\,\ell^{a}\bar m^{b}\,\ell^{c}\bar m^{d}
\;=\; -\tfrac12\!\left(h_{xx}-h_{yy}-2i\,h_{xy}\right),
\qquad
\Psi_0=\Psi_1=\Psi_2=\Psi_3=0.
\end{equation}
For an aligned Maxwell field in the Brinkmann background one likewise has a single NP Maxwell scalar:
\begin{equation}
\label{eq:app-phi2}
\phi_2 \;=\; F_{ab}\,\ell^{a}\bar m^{b}
=\tfrac{1}{\sqrt2}\left(F_{ux}+i\,F_{uy}\right),
\qquad \phi_0=\phi_1=0,
\end{equation}
which is the canonical “pure radiation” alignment.

\subsection{Petrov type and specialization to the present solutions}

From Eq.~\eqref{eq:app-Psi4} it follows immediately:

\begin{itemize}
\item \textbf{Type N (generic pp--wave).} If the tracefree Hessian \(\mathcal S_{ij}\) is not identically zero,
then \(\Psi_4\neq0\) while \(\Psi_0=\Psi_1=\Psi_2=\Psi_3=0\). The repeated principal null direction is
\(k=\partial_v\), and the spacetime is of Petrov type~N.

\item \textbf{Type O (conformally flat).} If \(\mathcal S_{ij}\equiv0\) (equivalently \(h_{xx}=h_{yy}\) and
\(h_{xy}=0\)), then \(\Psi_4=0\) and the Weyl tensor vanishes. The curvature is then purely Ricci (null radiation),
with no free gravitational field.
\end{itemize}

With the harmonic complement removed as in Remark~\ref{rem:Xi-coulomb-cancel}, the electromagnetic contribution to the
Brinkmann profile may be taken as the universal particular solution Eq.~\eqref{hEMsolution},
up to addition of independent vacuum pp--wave modes \(\Delta_\perp h_{\rm vac}=0\) which are set to zero by convention
when isolating the electrovac response. For this universal electrovac profile one has
\(
\mathcal H_{ij}=\partial_i\partial_j h_{\rm EM}
=-2\,\mathcal C\,\mathrm{tr}\,\mathbf D(u)\,\delta_{ij},
\
\Delta_\perp h_{\rm EM}=-4\,\mathcal C\,\mathrm{tr}\,\mathbf D(u),
\)
hence \(\mathcal S_{ij}\equiv0\) and the spacetime is \textbf{Petrov type O} (conformally flat), with nonzero Ricci
component \(R_{uu}\) describing aligned null radiation.

For the vacuum Brinkmann gravitational wave profile Eq.~\eqref{hG_quadratic}, one has \(\Delta_\perp h_{\mathrm{GW}}=0\) but
\(
\mathcal H_{ij}=\left(\begin{smallmatrix}H_{+}(u)&H_{\times}(u)\\ H_{\times}(u)&-H_{+}(u)\end{smallmatrix}\right)
\)
and therefore \(\mathcal S_{ij}=\mathcal H_{ij}\not\equiv0\) whenever \(H_{+}\) or \(H_{\times}\) is nonzero.
Thus \(h_{\mathrm{GW}}\) is \textbf{Petrov type N}, as expected for a plane gravitational pp--wave.

By Kerr--Schild linearity within the Brinkmann class, the total profile
\(h=h_{\rm EM}+h_{\rm GW}\) [cf.\ Eq.~\eqref{eq:h_total_general}] has Weyl tensor determined entirely by the tracefree
part of the transverse Hessian. In particular, with the universal electrovac choice above (type O) plus a nontrivial
vacuum GW (type N), the \textbf{combined} spacetime is \textbf{Petrov type N} wherever the GW amplitudes are nonzero.

\subsection{Scalar curvature invariants (VSI property)}

All polynomial scalar invariants formed from the Riemann tensor (and from the Weyl and Ricci tensors separately) vanish:
\begin{equation}
\label{eq:app-VSI}
R=0,\qquad R_{\mu\nu}R^{\mu\nu}=0,\qquad
R_{\mu\nu\rho\sigma}R^{\mu\nu\rho\sigma}=0,\qquad
C_{\mu\nu\rho\sigma}C^{\mu\nu\rho\sigma}=0,
\end{equation}
and likewise for complete contractions involving covariant derivatives of the curvature. This is the familiar VSI
(Vanishing Scalar Invariants) property of pp--waves.\cite{Griffiths2009,Stephani2003,ColeyHervikPravda2004VSI,MilsonMcNuttColey2013} A direct check uses
Eq.~\eqref{eq:app-Riemann-Ricci}: the only nonzero components carry at least one \(u\) index, and since \(g^{uu}=0\)
in Brinkmann form, every full contraction necessarily vanishes. For Einstein--Maxwell, the stress--energy is traceless
(\(T^\mu{}_\mu=0\)), hence \(R=0\), consistent with Eq.~\eqref{eq:app-VSI}.

\bibliographystyle{aipnum4-1}   
%

\end{document}